\documentclass[sigconf,nonacm]{acmart}
\pdfoutput=1

\usepackage{graphicx}
\usepackage{balance}
\usepackage{verbatim}
\usepackage{amsmath}
\usepackage{algorithmicx}
\usepackage[ruled]{algorithm}
\usepackage{algpseudocode}
\usepackage{pgfplots}
\usepackage{filecontents}
\usepackage{pgfplotstable}
\usepgfplotslibrary{groupplots}
\usepackage{url}
\usetikzlibrary{patterns,shapes,arrows,positioning,fit,decorations.pathreplacing}
\usepackage{xifthen}
\usepackage{xcolor}
\usepackage{mathtools}
\usepackage{tikz}
\usepackage{booktabs}
\usepackage{multirow}
\usepackage{listings}
\usepackage{color}
\usepackage{bbm}
\usepackage{dsfont}
\usepackage[center]{subfigure}
\usepackage{natbib}
\usepackage{tikzscale}
\usepackage{multicol}
\usepackage{cprotect}
\newcommand{\argmax}{\mathop{\rm argmax}}

\DeclareMathOperator{\reg}{\mathrm{Reg}}
\usetikzlibrary{pgfplots.groupplots}
\usetikzlibrary{calc}

\newcommand{\revision}[1]{\textcolor{black}{#1}}

\definecolor{dkgreen}{rgb}{0,0.6,0}
\definecolor{gray}{rgb}{0.5,0.5,0.5}
\definecolor{mauve}{rgb}{0.58,0,0.82}

\definecolor{bblue}{HTML}{4F81BD}
\definecolor{rred}{HTML}{C0504D}
\definecolor{ggreen}{HTML}{9BBB59}
\definecolor{ppurple}{HTML}{9F4C7C}
\definecolor{processblue}{cmyk}{0.96,0,0,0}

\newcommand{\alg}{\mathrm{Alg}}
\AtBeginDocument{%
  \providecommand\BibTeX{{%
    \normalfont B\kern-0.5em{\scshape i\kern-0.25em b}\kern-0.8em\asTeX}}}

\begin{document}


\title{UDO: Universal Database Optimization\\ using Reinforcement Learning}

\author{Junxiong Wang}
\email{junxiong@cs.cornell.edu}
\affiliation{%
\institution{Cornell University}
\city{Ithaca}
\state{NY}
\country{USA}
\postcode{14850}
}

\author{Immanuel Trummer}
\email{itrummer@cornell.edu}
\affiliation{%
  \institution{Cornell University}
  \city{Ithaca}
  \state{NY}
  \country{USA}
  \postcode{14850}
}

\author{Debabrota Basu}
\email{debabrota.basu@inria.fr}
\affiliation{%
  \institution{Scool, Inria Lille- Nord Europe}
  \city{Lille}
  \state{}
  \country{France}
  \postcode{59650}
}

\begin{abstract}
UDO is a versatile tool for offline tuning of database systems for specific workloads. UDO can consider a variety of tuning choices, reaching from picking transaction code variants over index selections up to database system parameter tuning. UDO uses reinforcement learning to converge to near-optimal configurations, creating and evaluating different configurations via actual query executions (instead of relying on simplifying cost models). To cater to different parameter types, UDO distinguishes heavy parameters (which are expensive to change, e.g. physical design parameters) from light parameters. Specifically for optimizing heavy parameters, UDO uses reinforcement learning algorithms that allow delaying the point at which the reward feedback becomes available. This gives us the freedom to optimize the point in time and the order in which different configurations are created and evaluated (by benchmarking a workload sample). UDO uses a cost-based planner to minimize reconfiguration overheads. For instance, it aims to amortize the creation of expensive data structures by  consecutively evaluating configurations using them. We evaluate UDO on Postgres as well as MySQL and on TPC-H as well as TPC-C, optimizing a variety of light and heavy parameters concurrently.
\end{abstract}

\maketitle

\section{Introduction}

We introduce \textit{UDO}, the \textit{Universal Database Optimizer}. UDO is an offline tuning tool that optimizes various kinds of tuning choices (e.g., physical design decisions as well as settings for database system configuration parameters), given an example workload and a tuning time limit. UDO does not rely on simplifying cost models to assess the quality of tuning options. Also, it does not require any kind of training data upfront. Instead, it relies only on feedback obtained via sample runs, after creating a tuning configuration to evaluate. This makes the optimization process expensive but avoids sub-optimal choices due to erroneous cost estimates, which are otherwise common~\cite{Borovica2012}. 

Given the tradeoff realized by UDO (i.e., high-quality, high-overhead optimization), we see two primary use cases. First, UDO is useful in scenarios where a configuration, obtained via expensive optimization, can be used over extended periods. This is possible if data and query workload properties do not change too frequently. Also, UDO is useful as an analysis tool for other tuning approaches. For instance, as UDO does not rely on cost or cardinality models, it can be used to uncover weaknesses in other recommender tools that are based on the latter. In this scenario, UDO adopts a similar role as previously proposed methods for query optimizer testing~\cite{Chaudhuri2009b, Trummer2019}, which generate guaranteed optimal plans via an expensive process (but are specific to query plans, as opposed to other tuning choices).


UDO operates on various types of tuning parameters, which are traditionally handled by separate tuning tools. For instance, in our experiments, we consider optimization of transaction query orders~\cite{Yan2016}, index selections~\cite{Ding2019, Chaudhuri2004}, as well as database system configuration parameters~\cite{Zhang1910, Zhang2019a}. Considering various parameter types together can be advantageous as optimal choices for one parameter type may depend on settings for other parameters (e.g., we may disable sequential scans, a configuration parameter, only if specific indexes are created). Hereby, we use the generic term \textbf{Parameter} for each tuning choice and the term \textbf{Configuration} for an assignment from parameters to values. UDO handles all parameters by a unified approach.

UDO explores the search space iteratively: selecting configurations to try, creating them (e.g., creating index structures or setting system parameters as specified by the configuration), and evaluating their performance on a workload sample. Evaluation is flexible to incorporate multiple metrics such as throughput or latency. We demonstrate optimization with both metrics on different database systems (Postgres and MySQL) and standard benchmarks (TPC-C and TPC-H). UDO uses \textit{Reinforcement Learning (RL)} to determine which configurations to try next. Improvement in performance measurements translate into reward values that guide an RL agent during search towards actions, i.e. configurations, that maximize the accumulated reward, i.e. the resulting performance. 

RL has been used previously for optimizing database system configuration parameters~\cite{Li2018, Zhang2019a} in particular. The main novelty of UDO lies in the fact that it broadens the scope of optimization to a much larger class of parameters. This becomes particularly challenging due to what we call \textbf{Heavy Parameters} (we distinguish them from \textbf{Light Parameters} in the following). For heavy parameters, it is expensive to change the parameter value. For instance, parameters that relate to index creations are expensive to change. Creating an index, in particular a clustered index, may take an amount of time that dominates query or transaction evaluation time for a small workload sample. Similarly, configuration parameters requiring a database server restart are relatively expensive to change. As we show in our experiments, a na\"ive RL approach is limited by costs of changing heavy parameters. This incurs high costs per iteration and slows down convergence.

UDO avoids this pitfall by giving heavy parameters special treatment. \textit{UDO separates heavy parameters from light ones and uses different reinforcement learning algorithms to optimize them}. Specifically for heavy parameters, it uses an RL algorithm that can adjust with delays until reward values for previous choices become available. We leverage such delayed feedbacks as follows. All configurations selected by the RL algorithm are forwarded to a \textit{planning component}. The planning component decides, when and in which order to create and to evaluate configurations. Depending on those choices, we are able to amortize cost for changing heavy parameters over the evaluation of many similar configurations. For instance, it allows us to create an expensive index once to evaluate multiple similar configurations that all include the index. The alternative approach, which is alternating between configurations that use or do not use the index requiring multiple index creations and drops, is less efficient.

For the current setting of heavy parameters, we use RL again to find optimal settings for light parameters. Of course, optimal settings for light parameters depend on the values for heavy parameters. UDO takes that into account and models the optimization of light parameters for each heavy parameter setting as a separate Markov Decision Process (MDP), to which an RL algorithm is applied to. In contrast to heavy parameters, we use a no-delay RL algorithm to converge faster to near-optimal settings for light parameters.

We propose a new Monte Carlo Tree Search (MCTS) variant, called delayed-Hierarchical Optimistic Optimization (HOO), that can be used for optimizing both, heavy and light parameters (with and without delays). We show that UDO converges to near-optimal configurations, given enough optimization time, when using that approach.


We demonstrate via experiments that the resulting system finds better configurations, compared to baselines, given the same amount of optimization time. We consider multiple standard benchmarks (TPC-H and TPC-C), multiple optimization metrics (throughput and latency), as well as different database management systems (Postgres and MySQL). In summary, our original, scientific contributions are the following.

\begin{itemize}
    \item We introduce an approach for optimizing various database tuning decisions using reinforcement learning. This approach is characterized by a factorisation of heavy and light parameters, the use of RL algorithms accepting delayed feedback, and a planner component that reduces re-configuration overheads by carefully planning evaluation orders.
    \item We experimentally demonstrate that the UDO system finds better configurations than baselines, given the same amount of optimization time. Our experiments cover various benchmarks and metrics.
    \item We propose a new MCTS variant, delayed-HOO, that can be used to optimize light and heavy parameters. We show that UDO converges to near-optimal solutions, using that approach, under moderately simplifying assumptions.
\end{itemize}



The remainder of this paper is organized as follows. First, in Section~\ref{sec:model}, we introduce our formal problem model and terminology used throughout the paper. Then, in Section~\ref{sec:overview}, we give a high-level overview of the UDO system. We analyze UDO formally in Section~\ref{sec:analysis}. \revision{In Section~\ref{sec:evaluation}, we describe mechanisms by which UDO evaluates batches of configurations efficiently. In Section~\ref{sec:rlalgs}, we introduce UDO's learning algorithms and analyze them formally in Section~\ref{sec:analysis}.} Then, in Section~\ref{sec:experiments}, we report results of our experimental evaluation. We discuss the related works in Section~\ref{sec:related} before concluding the paper.



\section{Formal Model}
\label{sec:model}

We introduce our problem model and associated terminology here.

\begin{definition}
A \textbf{Tuning Parameter} represents an atomic decision,  influencing performance of a database management system for a specific workload. It is associated with a (discrete) \textbf{Value Domain}, representing admissible parameter values. It may be subject to \textbf{Constraints}, restricting its values based on the values of other tuning parameters.
\end{definition}

We use the term ``parameter'' in a broad sense, encompassing system configuration parameter settings as well physical design decisions. In the following, we give examples for tuning parameters.

\begin{example}
Considering a set of candidate indices for a given database, we associate one tuning parameter with each candidate. Such index-related parameters have a binary value domain, representing whether the index is created or not. Equally, we can introduce a tuning parameter to represent the  \verb|random_page_cost| configuration parameter of the Postgres system (together with a set of values to consider). Finally, we may associate a query in a transaction template with a tuning parameter, representing the position within the template at which it is evaluated (the set of admissible positions is restricted via control flow and data dependencies).
\end{example}

\begin{definition}
Given fixed, ordered parameters, a \textbf{Configuration} $c$ is a vector, assigning a specific value to each parameter. The \textbf{Configuration Space} $C$ is the set of all possible configurations.
\end{definition}

Our goal is to find configurations that optimize a benchmark.

\begin{definition}\label{def:fmetric}
A \revision{\textbf{Benchmark Metric}} $f$ maps a configuration $c \in C$ to a real-valued performance result (i.e., $f:C\mapsto\mathbb{R}$)\revision{, which represents the performance of a configuration according to a specific metric for a specific benchmark}.
Higher performance results are preferable. \revision{We assume that $f$ is stochastic (i.e., evaluating the same configuration twice may not yield exactly the same performance).}
\end{definition}

\revision{Our definition of $f$ is deliberately generic, covering different types of benchmarks and metrics. A few examples follow.}


\begin{example}\label{ex:benchmarks}
In our experiments, we \revision{use the following two benchmark metrics among others. We consider a benchmark metric $f_1$ that maps configurations to the average throughput, measured over a fixed time period, when processing TPC-C transactions generated randomly according to a fixed distribution. Also, we consider a benchmark metric $f_2$ that maps configurations to a weighted sum between disk space $d$ consumed (e.g., for created indexes) and run time $t$ of all TPC-H queries (i.e., $f(c)=-d-\sigma\cdot t$ where $c$ is a configuration and $\sigma\in\mathbb{R}^+$ a user-defined scaling factor). Both benchmark metrics are implemented as a script (a black box from UDO's perspective) that returns a numerical performance result.}
\end{example}


We present a system, UDO, that solves the following problem.

\begin{definition}
An instance of \textbf{Universal Database Optimization} is characterized by a benchmark \revision{metric $f$} and a configuration space $C$. The goal is to find an optimal configuration $c^*$, \revision{maximizing the stochastic benchmark metric $f$ in expectation (i.e., $c^*=\argmax_{c\in C}\mathbb{E}[f(c)]$).} We denote the optimal expected performance (that of $c^*$) as $f^*$.
\end{definition}




For iterative approaches, the problem specification may also include a user-defined timeout for optimization. The qualification ``Universal'' attest to the fact that our approach is broadly applicable, in terms of parameter types, workloads, and performance metrics. We \revision{map} an UDO instance to multiple \revision{episodic} Markov Decision Processes, \revision{using the following definition}, and solve UDO using RL. 

\revision{\begin{definition}\label{def:emdp}
An \textbf{Episodic Markov Decision Process} (MDP) is defined by a tuple $\langle \mathcal{S},\mathcal{A},\mathcal{T},\mathcal{R},\mathcal{S}_D,\mathcal{S}_E\rangle$ where $\mathcal{S}$ is the state space, $\mathcal{A}$ a set of actions, and $\mathcal{T}:\mathcal{S}\times\mathcal{A}\mapsto\mathcal{S}$ a transition function linking state-action pairs to new states. $\mathcal{R}:\mathcal{S}\mapsto\mathbb{R}$ is a reward function mapping states to a reward value. We consider deterministic transitions but stochastic rewards. Optimization models an agent that performs steps. In each step, the agent selects an action, receives a reward, and transitions to the next state, based on the selected action. Optimization is divided into episodes. In each episode, the agent starts in state $\mathcal{S}_D\in\mathcal{S}$. The episode ends once it reaches one of the end states $\mathcal{S}_E\subseteq\mathcal{S}$. The goal is to find a policy (here: a sequence of actions as we consider deterministic transitions) that maximizes expected rewards per step.\end{definition}}

\revision{We introduce two scenario-specific instances of this formalism, associated with different parameter types.}

\begin{definition}\label{def:heavylight}
We distinguish \textbf{Heavy} and \textbf{Light Parameters}, based on the overheads associated with changing their values. Heavy parameters have high reconfiguration overheads, light parameters have negligible overheads. \revision{We denote by $C_H$ the configuration space for heavy parameters (i.e., a set of vectors representing all possible heavy parameter settings). We denote by $C_L$ the configuration space for light parameters. Hence, the entire configuration space $C$ can be written as $C=\{c_H\circ c_L|c_H\in C_H, c_L\in C_L\} = C_H \times C_L$ (assuming that heavy parameters are ordered before light parameters and writing vector concatenation as $\circ$).}
\end{definition}

\revision{Currently, UDO considers all parameters representing physical data structures such as indexes as heavy (as changing such parameters means creating or dropping the associated data structure). Also, UDO considers parameters as heavy that require a database server restart to make value changes effective. The other parameters are considered light.}


\revision{\begin{definition}
For the \textbf{Heavy Parameter MDP}, states correspond to configurations for heavy parameters (i.e., $\mathcal{S}\subseteq C_H$) and each action changes one heavy parameter to a new value (i.e., an action is defined as a pair $\langle p,v\rangle$ representing parameter $p$ and new value $v$). The transition function maps a configuration (i.e., state) with a parameter value change (i.e., action) to a new configuration, reflecting the changed value. The start state $\mathcal{S}_D\in C_H$ represents the default configuration (i.e., no created indices and default values for all system parameters). All states reachable from the start state with a given number of actions are end states (we typically use a threshold of four actions). The reward function $\mathcal{R}$ is scenario-specific and based on the benchmark metric $f$. The reward for a state representing heavy parameter configuration $c_H$ is proportional to $\argmax_{c_L\in C_L}f(c_H\circ c_L)$, i.e.\ to the value of the benchmark metric when combining $c_H$ with the best possible configuration $c_L$ for light parameters. We scale raw rewards by subtracting rewards for the default configuration (e.g., if $f$ measures throughput for a specific benchmark, UDO considers the throughput improvement compared to default settings as reward function).
\end{definition}}

\revision{The definition above uses optimal configurations for light parameters, leading to the second MDP version.}

\revision{\begin{definition}
A \textbf{Light Parameter MDP} $\mathcal{M}_L[c_h]$ is introduced for each heavy parameter configuration $c_h$ (in practice, we limit ourselves to configurations explored by UDO). Its states represent configurations for light parameters, its actions represent value changes for light parameters (analogue to the previous definition). The start state represents default values for all light parameters and end states are defined by a fixed number of light parameter changes, compared to the default. The reward function $\mathcal{R}_L[c_H]$ is defined as $\mathcal{R}_L[c_H](c_L)=f(c_H\circ c_L)-f_D$ where $f_D$ is performance of the default configuration.
\end{definition}}

As shown above, we model optimization for light parameters as a family of MDPs, where each MDP is associated with a specific heavy parameter configuration. Our problem model assumes that only the reward function (i.e., performance) depends on heavy parameters. It is possible to extend the definition above to cover cases where admissible values for light parameters depend on currently chosen heavy parameters. In that case, states, transitions, and actions must be instantiated for specific heavy parameter settings as well. The following example illustrates the interplay between heavy and light parameter MDPs.


\tikzstyle{hstate}=[rectangle, fill=blue!10]
\tikzstyle{lstate}=[rectangle, fill=green!10]
\tikzstyle{transition}=[->, black, thick]
\tikzstyle{mapping}=[->, black, thick, dashed]

\begin{figure}
    \centering
    \begin{tikzpicture}
        \node[hstate] (root) at (0,0) {$\langle0,0,?\rangle$};
        \node[hstate] (hstate1) at (-3,0) {$\langle1,0,?\rangle$};
        \node[hstate, fill=red!50] (hstate2) at (3,0) {$\langle0,1,?\rangle$};
        \node[lstate] (l1) at (-3,-1) {$\langle1,0,2MB\rangle$};
        \node[lstate] (l2) at (-3,-2) {$\langle1,0,10MB\rangle$};
        \node[lstate, fill=red!50] (l3) at (-1,-1) {$\langle1,0,25MB\rangle$};
        \node[lstate] (l4) at (-1,-2) {$\langle1,0,50MB\rangle$};
        \node[lstate] (l5) at (3,-1) {$\langle0,1,2MB\rangle$};
        \node[lstate] (l6) at (3,-2) {$\langle0,1,10MB\rangle$};
        \node[lstate] (l7) at (1,-1) {$\langle0,1,25MB\rangle$};
        \node[lstate, fill=red!50] (l8) at (1,-2) {$\langle0,1,50MB\rangle$};
        \draw[transition] (root) -- (hstate1);
        \draw[transition] (root) -- (hstate2);
        \draw[transition] (l1) -- (l2);
        \draw[transition] (l1) -- (l3);
        \draw[transition] (l1) -- (l4);
        \draw[transition] (l5) -- (l6);
        \draw[transition] (l5) -- (l7);
        \draw[transition] (l5) -- (l8);
        \draw[mapping] (hstate1) -- (l1);
        \draw[mapping] (hstate2) -- (l5);
    \end{tikzpicture}
    \caption{Extract of heavy (top) and light (bottom) parameter MDPs for a space with two heavy and one light parameter. Optimal states for each MDP are marked up in red.}
    \label{fig:mdp}
\end{figure}
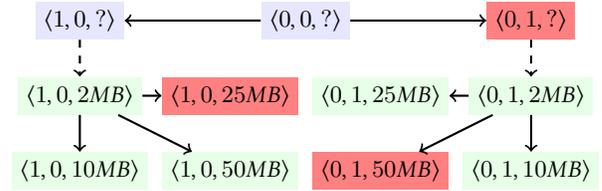

\begin{example}\label{ex:mdps}
Figure~\ref{fig:mdp} illustrates \revision{part of a} two-level MDP. In the illustrated scenario, the configuration space contains two heavy parameters (e.g., index creation decisions) as well as one light parameter (e.g., the maximal amount of main memory used per operator). Rectangles represent states in the figure and are annotated with configuration vectors (reporting parameter in the aforementioned order). The upper part of the figure illustrates the heavy parameter MDP (note that the light parameter is not specified). Solid lines mark transitions due to actions changing the configuration. Dashed lines mark mappings between heavy parameter states and the start state of the associated light parameter MDP. \revision{As index-related parameters are binary, the heavy parameter MDP has four states (out of which three are shown).} In total, the figure illustrates three MDPs (the one for heavy parameters and light parameter MDPs for two heavy configurations). Optimal states are marked up in red, showing that the optimal settings for light parameters may depend on the heavy parameter settings. Identifying optimal settings for heavy parameters requires obtaining optimal, associated settings for light parameters first.
\end{example}

\section{System Overview}
\label{sec:overview}

\tikzstyle{processstep}=[fill=blue!10, draw, minimum width=6cm]
\tikzstyle{processflow}=[draw, ultra thick, ->]
\tikzstyle{datasent}=[font=\itshape]
\tikzstyle{complabel}=[red, font=\bfseries]

\begin{figure}[t]
    \centering
    \begin{tikzpicture}
    
        \draw[rounded corners, fill=gray!5] (-3,-0.5) ++ (-0.5, 0.5) rectangle ($ (3,-5.5) + (0.3,-0.1) $);
        \node[font=\bfseries] at (0,-0.25) {Universal Database Optimizer};
        \node[datasent] (udoinput) at (0,0.5) {Benchmark Metric, Configuration Space, Time Budget};
        \node[datasent] (udooutput) at (0,-6.1) {Best Configuration};
        \draw[processflow] (udoinput.south) -- ++ (0,-0.2);
        \draw[processflow] (udooutput.north) ++ (0,0.2) -- ++ (0,-0.2);
    
        \node[processstep] (rlheavy) at (0,-1) {RL for Heavy Parameters};
        \node[processstep] (pick) at (0,-2) {Pick Configurations to Evaluate};
        \node[processstep] (order) at (0,-3) {Order Selected Configurations};
        \node[processstep] (rllight) at (0,-4) {RL for Light Parameters};
        \node[processstep] (dbms) at (0,-5) {Evaluate Configuration};
        
        \draw[processflow] (rlheavy.south east) -- (pick.north east);
        \draw[processflow] (pick.south east) -- (order.north east);
        \draw[processflow] (order.south east) -- (rllight.north east);
        \draw[processflow] (rllight.south east) -- (dbms.north east);
        \draw[processflow] (dbms.north west) -- (rllight.south west);
        \draw[processflow] (dbms.north west) -- ++ (-0.2,0) -- ++ (0,3.74) -- (rlheavy.west);
        
        \node[datasent] at (0,-1.5) {Configuration for Heavy Parameters+Deadline};
        \node[datasent] at (0,-2.5) {Configurations for Heavy Parameters (Set)};
        \node[datasent] at (0,-3.5) {Configurations for Heavy Parameters (List)};
        \node[datasent, anchor=east] at (3,-4.5) {Configuration for All Parameters};
        \node[datasent, anchor=west] at (-3,-4.5) {Result};
        
        \node[complabel] at (2.5,-1) {(A)};
        \node[complabel] at (2.5,-2) {(B)};
        \node[complabel] at (2.5,-3) {(C)};
        \node[complabel] at (2.5,-4) {(D)};
        \node[complabel] at (2.5,-5) {(E)};

    \end{tikzpicture}
    \caption{\revision{Overview of UDO system (rectangles represent processing steps, arrows represent data flow).}}
    \label{fig:process}
\end{figure}
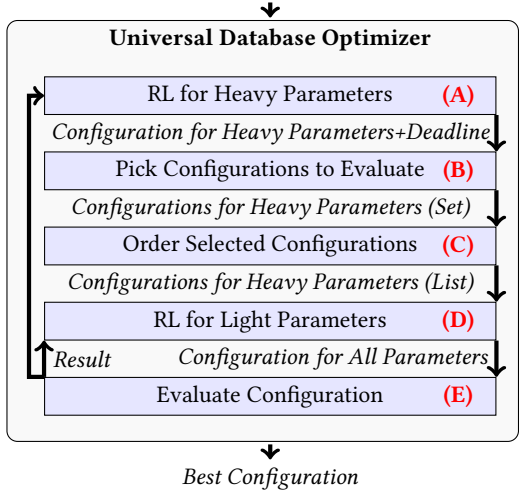


\revision{Figure~\ref{fig:process} shows an overview of UDO and the interplay between its components. The input to UDO is a benchmark metric to optimize, a configuration space, and an optimization time budget. The configuration space is specified as a set of index candidates to consider, a set of database system parameters with alternative values to try, and (optionally) a set of alternative versions for each query or transaction template. UDO considers index parameters as heavy and the others as light. The output is the best configuration found until the time limit.}


\revision{UDO iterates until the time limit is reached\footnote{\revision{Instead of a fixed optimization time budget, we could use other termination criteria as well. For instance, the algorithm could terminate once a given number of iterations does not yield improvements above a configurable threshold.}}. In each iteration, UDO first chooses a configuration of heavy parameters to explore (Component~A). UDO uses reinforcement learning for this decision, balancing the need for exploration (analyzing configurations about which little information is available) and exploitation (refining configurations that seem to well) in a principled manner. Evaluating a new configuration for heavy parameters can however be expensive. It involves changing the current database configuration to the one to evaluate, e.g.\ by creating indexes. Doing so becomes cheaper if the current configuration is close to the one to evaluate. Hence, UDO tries to optimize the point in time at which heavy parameter configurations are evaluated. UDO uses a specialized reinforcement learning algorithm that does not expect evaluation results immediately after selecting a configuration. Instead, it allows for a certain delay (measured as the number of iterations between selection and evaluation result). Selected configurations for heavy parameters are added to a buffer, associated with a deadline until which the result must become available. Note that the learning algorithm does not  consider the current database state for \textit{deciding which configuration to explore} (doing so may prevent UDO from finding promising configurations that are far from the current one). Instead, it merely creates opportunities for cost reductions by other system components.}

\revision{In each iteration, UDO selects a set of heavy parameter configurations to evaluate from the aforementioned buffer (Component B). Configurations are selected if either their deadline has been reached (in this case, there is no choice) or if their evaluation is cheaper than usual (e.g., because they share indexes with configurations that must be evaluated). Selected configurations are ordered for evaluation (Component C). The goal of evaluation is to reduce reconfiguration cost by placing similar configurations consecutively. For instance, if configurations with similar indexes are evaluated consecutively, some index creation cost can be amortized.}

\revision{Next, UDO selects values for light parameters (Component D). The best configuration for light parameters may depend on the heavy parameter configuration. For instance, we may want to enable or disable specific join algorithms (by setting parameters such as enable\_nestloop for Postgres), depending on which indexes are available. For specific configurations of heavy parameters, UDO learns suitable settings for light parameters via reinforcement learning. Here, reconfiguration is cheap. Hence, UDO uses a standard reinforcement learning algorithm without delays. Light parameters are optimized for the current heavy configuration for a fixed number of iterations of the latter learning algorithm. Note that statistics for light parameters are saved and will be used as starting point if the same heavy configuration is selected again. A fully specified configuration (i.e., for light and heavy parameters) is evaluated via the benchmark metric. This involves executing a script that executes a sample workload and returns the performance metric to optimize. The evaluation results are used to update the statistics of the two learning algorithms (Components D and A).}

\begin{algorithm}[t!]
\caption{UDO main function.}
\label{alg:twoLevelVar}
\renewcommand{\algorithmiccomment}[1]{// #1}
\begin{small}
\begin{algorithmic}[1]
\State \textbf{Input:} Benchmark \revision{metric $f$}, configuration space $C$, RL algorithms $\alg_H$ and $\alg_L$ for heavy and light parameter optimization
\State \textbf{Output:} a suggested configuration for best performance
\Function{UDO}{\revision{$f$}$, C, \alg_H, \alg_L$}
\State \Comment{Divide into heavy ($C_H$) and light ($C_L$) parameters}
\State $\langle C_H, C_L \rangle \gets$ \Call{SSA.SplitParameters}{$C$}
\State \Comment{Until optimization time runs out}
\For{$t\gets1,\ldots,\alg_H.Time$}
\State \Comment{Select next heavy parameter configuration}
\State $c_{H,t}\gets $\Call{RL.Select}{$\alg_H,C_H, c_{H, t-1}$}
\State \Comment{Submit configuration for evaluation}
\State \Call{EVAL.Submit}{$c_{H,t},t+\alg_H.maxDelay$}
\State \Comment{Receive newly evaluated light configurations}
\State $E\gets$\Call{EVAL.Receive}{$\alg_L,$\revision{$f$}$,C_L,t$}
\State \Comment{Update statistics for heavy parameters}
\State \Call{RL.Update}{$\alg_H,E$}
\EndFor
\State \Return best obtained configuration
\EndFunction
\end{algorithmic}
\end{small}
\end{algorithm}

\revision{Algorithm~\ref{alg:twoLevelVar} describes the main loop, executed by UDO, in more formal detail. Beyond benchmark metric and configuration space, it obtains two parameters specifying hyper-parameters for the two reinforcement learning algorithms used. These include optimization time as well as other parameters (e.g., the maximal amount of allowed delay) whose impact we analyze in Section~\ref{sec:experiments}. Users only need to specify optimization time while defaults are available for the other algorithm parameters (hence, Figure~\ref{fig:process} only references the former parameter).}

\revision{Algorithm~\ref{alg:twoLevelVar} first classifies parameters as heavy or light (Line~5). We use a simple heuristic and classify parameters requiring index creations or database server restarts as heavy, the other ones as light. Next, Algorithm~\ref{alg:twoLevelVar} iterates until optimization time runs out. It selects interesting heavy parameter configurations to evaluate via reinforcement learning (Line~9). It submits requests for evaluation, setting a deadline until which the result must be available (Line~11). It receives evaluation results for previously submitted requests (potentially, but not necessarily, including the one submitted in the current iteration). The results are used to update the statistics for learning (Line~15). Finally, the best configuration is returned.}

\revision{We discuss the sub-functions related to evaluating configurations (Functions~\textproc{EVAL.Submit} and \textproc{EVAL.Receive}) in Section~\ref{sec:evaluation}. In Section~\ref{sec:rlalgs}, we discuss the learning algorithms used (Functions~\textproc{RL.Select} and \textproc{RL.Update}).}

\section{Evaluating Configurations}
\label{sec:evaluation}

\begin{algorithm}[t!]
\caption{EVAL: Functions for evaluating configurations.}
\label{alg:evalVar}
\renewcommand{\algorithmiccomment}[1]{// #1}
\begin{small}
\begin{algorithmic}[1]
\State \Comment{Global variable representing evaluation requests}
\State $R\gets\emptyset$
\vspace{0.15cm}
\State \textbf{Input:} heavy configuration $c_H$ to evaluate and time $t$
\State \textbf{Effect:} adds new evaluation request 
\Procedure{EVAL.Submit}{$c_H,t$}
\State $R\gets R\cup\{\langle c_H,t\rangle\}$
\EndProcedure
\vspace{0.15cm}
\State \textbf{Input:} RL algorithm $\alg_L$, benchmark \revision{metric $f$}, time $t$, and space $C_L$
\State \textbf{Output:} evaluated configurations with reward values
\Function{EVAL.Receive}{$\alg_L,$\revision{$f$}$,C_L,t$}
\State \Comment{Choose configurations from $R$ to evaluate now}
\State $N\gets$\Call{PickConf}{$R,t$}
\State \Comment{Remove from pending requests}
\State $R\gets R\setminus N$
\State \Comment{Prepare evaluation plan}
\State $P\gets$\Call{PlanConf}{$N$}
\State \Comment{Collect evaluation results by executing plan}
\State $E\gets\emptyset$
\For{$s\in P.steps$}
\State \Comment{Prepare evaluation of next configurations}
\State \Call{ChangeConfig}{$s.hconf$}
\State \Comment{Find (near-)optimal light parameter settings}
\State $c_L\gets$\Call{RL.Optimize}{$\alg_L,s.hconf,C_L,$\revision{$f$}}
\State \Comment{Take performance measurements on benchmark}
\State $b\gets$\Call{Evaluate}{\revision{$f$}$,s.hconf,c_L$}
\State \Comment{Add performance result to set}
\State $E\gets E\cup\{\langle c_L,s.hconf,b\rangle\}$
\EndFor
\State \Comment{Return evaluation results}
\State \Return{$E$}
\EndFunction
\end{algorithmic}
\end{small}
\end{algorithm}

\revision{We discuss how configurations are selected and ordered for evaluation. In Section~\ref{sub:evaluationsoverview}, we describe the implementation of the evaluation functions invoked by Algorithm~\ref{alg:twoLevelVar}. In Section~\ref{sub:picking}, we describe how configurations to evaluate are selected (Component~B in Figure~\ref{fig:process}). In Section~\ref{sub:ordering}, we discuss the method used to order configurations to evaluate for minimal cost.}

\subsection{Evaluation Overview}
\label{sub:evaluationsoverview}


The evaluation interface offers two functions, represented in Algorithm~\ref{alg:evalVar} \revision{(and used in Algorithm~\ref{alg:twoLevelVar})}. First, it accepts evaluation requests (\textproc{EVAL.Submit}), allowing to submit configurations for evaluation, together with an evaluation deadline. Second, it allows triggering evaluations via the \textproc{EVAL.Receive} function. 
Algorithm~\ref{alg:evalVar} maintains a global variable $R$ (whose state persists across different calls to the two interface functions). This variable contains pending requests for evaluating specific configurations. Each configuration to evaluate is only partially specified (i.e., it assigns values to a subset of configuration parameters). More precisely, configurations to evaluate only contain specific values for heavy parameters. \revision{During evaluation, we learn suitable values for light parameters} to accurately assess the potential of the heavy parameter settings. Items in $R$ correspond to tuples, combining a heavy parameter configurations with an evaluation deadline. This deadline specifies the latest possible time (measured as the number of main loop iterations, as per Algorithm~\ref{alg:twoLevelVar}) at which evaluation results must be generated. The submission function (Procedure~\textproc{EVAL.Submit}) simply adds one more tuple to set variable $R$.


Calling Function~\textproc{EVAL.Receive} triggers evaluation of a subset of pending configurations. It is up to \revision{that function} itself to choose, within certain boundaries, the set of configurations to evaluate. The function returns the results of those evaluations. As input, Function~\textproc{EVAL.Receive} obtains a configuration parameter $\alg_L$, specifying the algorithm to use for optimizing light parameters (for fixed  heavy parameter values). Also, it receives the benchmark \revision{metric $f$}, the space of light \revision{configurations} $C_L$, and the current time $t$ as input. The latter is important to decide which configurations must be evaluated in the current invocation.

As a first step (Line~12), Function~\textproc{EVAL.Receive} determines the set of configurations to evaluate in the current invocation. If the time $t$ has reached the deadline of any pending configurations, those configurations must be included in that set. For other configurations, the evaluator can choose to evaluate them now or to postpone evaluation. \revision{We describe the selection mechanisms in Section~\ref{sub:picking}.} Having selected configurations to evaluate, the algorithm removes those configurations from the pending set $R$. 


Having selected a set of configurations, Algorithm~\ref{alg:evalVar} decides how to evaluate them. Function~\textproc{PlanConf} selects a plan to evaluate the given set of configurations. Evaluating configurations in the right order can save significant overheads, compared to a random permutation. In particular, ordering them allows to amortize re-configuration overheads (e.g., overheads for creating an index) over the evaluation of multiple, similar configurations. The planner function (\textproc{PlanConf}) exploits this fact and aims at minimizing cost. \revision{We describe the planning mechanism in Section~\ref{sub:ordering} in detail.} 


After selecting a plan, Algorithm~\ref{alg:evalVar} processes the plan steps in order (loop from Line~19 to Line~28). For each plan step $s$, the system first executes re-configuration actions required to evaluate specific heavy parameter settings (Line~21). Then, it selects a (near-)optimal setting of light parameters, specifically for the current configuration of heavy parameters (Line~23). Here, we invoke a reinforcement learning algorithm described via tuning parameters $\alg_L$. We discuss learning algorithms to implement this step \revision{in Section~\ref{sec:rlalgs}}. Finally, Algorithm~\ref{alg:evalVar} benchmarks the current heavy and light parameter setting (Line~25) and adds the result to the set (Line~27). All evaluation results are ultimately returned to the invoking function (Line~30).

\subsection{Picking Configurations to Evaluate}
\label{sub:picking}

\begin{algorithm}[t!]
\caption{\textproc{PickConf}: Methods for picking configurations to evaluate.}
\label{alg:pickConfigurations}
\renewcommand{\algorithmiccomment}[1]{// #1}
\begin{small}
\begin{algorithmic}[1]
\State \textbf{Input:} Evaluation requests $R$, current timestamp $t$
\State \textbf{Output:} Set of configurations to evaluate
\Function{PickConf-Threshold}{$R,t$}
\State \Comment{Was size threshold reached?}
\If{$|R|\geq\rho$}
\State \Comment{Return all requests}
\State \Return{$R$}
\Else
\State \Return{$\emptyset$}
\EndIf
\EndFunction
\vspace{0.15cm}
\State \Comment{Initialize maximal cost savings for each request}
\State $S=\emptyset$
\vspace{0.15cm}
\State \textbf{Input:} Evaluation requests $R$, current timestamp $t$
\State \textbf{Output:} Set of configurations to evaluate
\Function{PickConf-Secretary}{$R,t$}
\State \Comment{Add requests whose deadline is reached}
\State $E\gets \{\langle c_H,t_D\rangle\in R|t_D\geq t\}$
\State \Comment{Remove requests from pending set}
\State $R\gets R\setminus E$
\State \Comment{Iterate over requests}
\For{$r=\langle c_H,t_D\rangle\in R$}
\State \Comment{Calculate re-configuration cost savings}
\State $s\gets$\Call{CostSavings}{$r,E$}
\State \Comment{Retrieve maximal savings so far}
\State $m\gets S(r)$
\State \Comment{Should we evaluate?}
\If{$t-(t_D-\delta)\geq\delta/e\wedge s>m$}
\State $E\gets E\cup\{r\}$
\EndIf
\State \Comment{Update maximally possible savings}
\State $S(r)\gets\max(m,s)$
\EndFor
\State \Return{$E$}
\EndFunction
\end{algorithmic}
\end{small}
\end{algorithm}


We present two strategies for selecting configurations to evaluate \revision{(invoked in Line~12 of Algorithm~\ref{alg:evalVar} and represented as Component~B in Figure~\ref{fig:process})}. Algorithm~\ref{alg:pickConfigurations} shows corresponding pseudo-code. The two functions represented in Algorithm~\ref{alg:pickConfigurations} (Function~\textproc{PickConf-Threshold} or \textproc{PickConf-Secretary}) implement the call in Line~12 of Algorithm~\ref{alg:evalVar}. Next, we discuss the two strategies in more detail.

The first strategy, represented by Function~\textproc{PickConf-Threshold}, is relatively simple. We select all pending evaluation requests for processing if their number has reached a threshold. This threshold is represented as parameter $\rho$ in the pseudo-code. Before reaching the threshold, we simply collect evaluation requests without actually processing them (i.e., the set of selected requests is empty). By evaluating requests in batches, we hope to amortize re-configuration overheads via the planning mechanisms outlined in the next \revision{subsection}. The threshold $\rho$ is a tuning parameter. It is associated with a tradeoff. Choosing $\rho$ too small reduces chances for cost amortization. Choosing $\rho$ too large means that we introduce significant delays for the RL algorithm between a configuration is selected and evaluated. Delaying feedback may increase time spent in exploring uninteresting parts of the search space. Note that $\rho$ must be smaller or equal to the maximal delay, allowed by the RL algorithm. In our experiments, we typically set $\rho$ to 20. Empirically, we determined this setting to work well for many scenarios. 

Our second strategy, written as Function~\textproc{PickConf-Secretary}, is more sophisticated and often works better in practice. It is motivated by algorithms for solving the so called ``Secretary Problem''~\cite{Hill2009}. This problem models a job interview for a single position, in which a hiring decision must be made directly after each interview. This decision is hard due to uncertainty with regards to the quality of the remaining candidates. A popular algorithm for this problem reviews a fraction of $1/e$ of candidates without hiring any. Then, it selects the first candidate better than all previously seen candidates (or the last candidate, if no such candidate emerges). It can be shown that this strategy makes a near-optimal choice likely. We use an adaption of this algorithm for our problem.

In our case, candidates correspond to evaluation times for a fixed configuration. The re-configuration cost, required to test a specific configuration, decreases if similar configurations were evaluated before. E.g., we do not have to create an expensive index, part of a configuration to evaluate, if that index was created before. So, instead of immediately evaluating a configuration, we may want to wait until similar configurations are requested. Of course, we cannot know precisely which configurations will be submitted for evaluation in the future. This is akin to the uncertainty about the quality of future job candidates. 

Algorithm~\ref{alg:pickConfigurations} keeps track of possible cost savings for specific configurations. Global variable $S$ keeps track of maximal savings in re-configuration costs for specific configurations, over different invocations of \textproc{PickConf-Secretary}. We compare current cost savings to the maximum seen so far to decide when to evaluate. Intuitively, we want to evaluate configurations in invocations, during which we can obtain particularly high cost savings. 

Function~\textproc{PickConf-Secretary} first selects all evaluation requests whose evaluation deadline has been reached (Line~18). Then, we iterate over the remaining requests. For each request, we calculate re-configuration cost savings, assuming that we evaluate it after the configurations selected for evaluation so far. Next, we retrieve maximal cost savings observed for this configuration so far (Line~26). We select the configuration for evaluation if we have observed cost savings over a sufficiently large period (condition $t-(t_D-\delta)\geq \delta/e$ where $\delta$ is the maximal delay and $e$ Euler's number) and if current savings exceed the previous optimum (condition $s>m$). 

\subsection{Optimizing Evaluation Order}
\label{sub:ordering}

Given a set of configurations to evaluate, \revision{we re-order them to minimize evaluation overheads (invoked in Line~16 of Algorithm~\ref{alg:evalVar} and represented as Component~C in Figure~\ref{fig:process}). The following example illustrates the principle.} 


\begin{example}
We describe configurations by vectors in which each vector component represents a parameter value. Assume we have to evaluate configurations $(1,1,16MB)$, $(0,0,12MB)$, and $(0,1,16MB)$. Here, the first two components indicate whether two specific indexes are created or not, the third component represents the (configurable) amount of working memory. Assume that the latter parameter requires a server restart with a duration of 10 seconds to take effect. For simplicity, we assume that creating an index takes 20 seconds while dropping one is free. Evaluating the configurations in the given order creates (pure configuration switching) overheads of $2\cdot 20+10+10+20+10=90$ seconds (assuming that no indexes are initially created and an initial setting of $8MB$ for memory). If we evaluate them in the order $(0,0,12MB)$, $(0,1,16MB)$, and $(1,1,16MB)$ instead, those overheads reduce to $10+10+20+20=60$ seconds. Relative savings tend to increase with the size of configuration batches.
\end{example}

\revision{We introduce the associated optimization problem formally.}

\begin{definition}\label{def:orderingproblem}
An instance of \textbf{Reconfiguration Cost Minimization} is defined by a set $R=\{r_i\}$ of requested (heavy parameter) configurations to evaluate and a cost function $\mathcal{C}:R\times R\mapsto\mathbb{R}^+$ that maps a pair $r_1$, $r_2$ of requested configurations to the cost for switching from $r_1$ to $r_2$ (e.g., by creating indexes that appear in $r_2$ but not in $r_1$). A solution is a permutation $\Pi:\mathbb{N}\mapsto R$ of configurations, representing evaluation order with cost $\sum_i\mathcal{C}(\Pi(i),\Pi(i+1))$. An optimal evaluation order minimizes cost.
\end{definition}





\revision{In the current implementation, we approximate $\mathcal{C}(r_1,r_2)$ by only considering indexes that appear in $r_2$ but not $r_1$ and summing up the cardinality of the indexed table over all added indexes.} Next, we analyze the computational complexity of this problem (called ``reconfiguration cost minimzation'' in the following).

\begin{theorem}
Reconfiguration cost minimization is NP-hard.
\end{theorem}
\begin{proof}
\label{pf:reconfiguration}
Consider an instance of the Hamiltonian graph problem. This instance is described by a graph $G$, the goal is to construct a path visiting each node once. We reduce to reconfiguration cost minimization as follows. For each node $i$ in $G$, we introduce one evaluation request $r_i$. For each pair of nodes $i$ and $j$, connected by an edge in $G$, we set the reconfiguration cost $\mathcal{C}(r_i,r_j)$ to zero, otherwise to one. Assume we find an evaluation order with a reconfiguration cost of zero. In this case, we obtain a Hamiltonian path in the original problem instance (visiting nodes, associated with requests, in the order in which requests are selected for evaluation). As each request is evaluated once, the associated graph node is visited once. As the reconfiguration cost is zero, all visited nodes are connected by edges.
\end{proof}


Hence, we must choose between efficient optimization and guaranteed optimal results. In the following, we present a greedy and an exhaustive algorithm to solve this problem.


\begin{algorithm}[t!]
\caption{\textproc{PlanConf}: Order configurations for evaluation.}
\label{alg:greedyPlan}
\renewcommand{\algorithmiccomment}[1]{// #1}
\begin{small}
\begin{algorithmic}[1]
\State \textbf{Input:} Evaluation requests $R$
\State \textbf{Output:} Requests in suggested evaluation order
\Function{PlanConf-Greedy}{$R$}
\State \Comment{Initialize list of ordered requests}
\State $O\gets[]$
\State \Comment{Iterate over all requests}
\For{$r\in R$}
\State \Comment{Find optimal insertion point}
\State $i\gets\arg\min_{i\in 0,\ldots,|O|}C_R(O[i-1],O[i])+C_R(O[i],O[i+1]))$
\State \Comment{Insert current request there}
\State $O.insert(i,r)$
\EndFor
\State \Return{$O$}
\EndFunction
\end{algorithmic}
\end{small}
\end{algorithm}

Algorithm~\ref{alg:greedyPlan} generates evaluation orders via a simple, greedy approach. The input is a set of evaluation requests (each one referencing a configuration to evaluate). Starting from an empty list, we expand the evaluation order gradually, by adding one more request in each iteration. We insert each request greedily at the position where it leads to minimal reconfiguration overheads. We measure re-configuration overheads via function $C_R(c_1,c_2)$, measuring reconfiguration overheads to move from configuration $c_1$ to configuration $c_2$. Those overheads include for instance index creation overheads for indices that appear in $c_2$ but not in $c_1$. After identifying the position with minimum overheads, we expand the order accordingly.


Next, we show how to transform the problem of ordering evaluations into an integer linear program. After doing so, we can use corresponding solvers to find an optimal solution quite efficiently. Our decision variables are binary: we introduce variables $e_t^r$ to indicate whether request $r$ is evaluated at time $t$. We introduce variables for each request $r\in R$ to evaluate and for $|R|$ time steps. We evaluate one configuration at each time step, represented by constraints of the form $\sum_{r}e_t^r=1$ (for each time step $t$). Also, we must evaluate each configuration once which we represent by the constraint $\sum_{t}e_t^r=1$ (for each request $r$)\footnote{Strictly speaking, the last constraint is redundant as we evaluate exactly one configuration in each times step.}. The objective function is determined by reconfiguration costs. For each pair of configuration requests $r_1$ and $r_2$, we can estimate reconfiguration cost $C_R(r_1,r_2)$ by comparing the associated configurations. We introduce binary variables of the form $i_t^{r_1,r_2}$, indicating whether reconfiguration costs for moving from $r_1$ to $r_2$ is incurred at time $t$. We introduce those variables for each pair of configurations and for each time step. The objective function is given as $\sum_{t,r_1,r_2}c_R(r_1,r_2)\cdot i_t^{r_1,r_2}$ (our goal is to minimize this function). Lastly, we need to ensure that the value assignments for variables $i_t^{r_1,r_2}$ and $e_t^r$ are consistent. Due to the objective function, variables $i_t^{r_1,r_2}$ will be set to zero if possible. Hence, we only must constrain them to one if the context requires it. We do so by introducing constraints of the form $i_t^{r_1,r_2}\geq (e_t^{r_1}+e_{t+1}^{r_2})/2$ for each pair of requests and for each time step. The optimal solution to this linear program describes an optimal evaluation order.

\section{Reinforcement Learning}\label{sec:rlalgs}

UDO uses RL algorithms from the family of Monte Carlo Tree Search (MCTS)~\citep{Coquelin2007a} methods. UDO can be instantiated with different algorithms, for optimizing light and heavy parameters respectively. Our implementation supports multiple algorithms as well. We discuss some of them in the following.

Throughout the pseudo-code presented so far, we used three sub-functions that relate to RL: \textproc{RL.Select}, \textproc{RL.Update}, and \textproc{RL.Optimize}. \revision{Those functions were used in Algorithm~\ref{alg:twoLevelVar} and \ref{alg:evalVar}.} The implementation of those functions depends on the RL algorithm used (as indicated by the $\alg$ parameter). The first function, \textproc{RL.Select}, selects the next action to take, based on algorithm-specific statistics. The second function, \textproc{RL.Update}, updates those statistics based on feedback. Function~\textproc{RL.Optimize} is based on the latter two functions and invokes them repeatedly for optimization.

Next, we show how to implement those functions for one specific RL algorithm. This algorithm follows the \textit{Hierarchical Optimistic Optimization} (HOO)~\citep{bubeck2011a} framework, a generalized version of the well-known UCT  algorithm~\citep{Kocsis2006}. We extend that algorithm with a mechanism for accepting delayed feedback. We call this algorithm \textit{Delayed Hierarchical Optimistic Optimization} (Delayed-HOO). This is the algorithm used for our experiments for optimizing both, heavy and light parameters (when optimizing light parameters, we set the allowed delay to zero). While based closely on existing components for action selection~\citep{ucbv} and delayed feedback management~\cite{Joulani2013a}, the combination of those components is novel. 

First, we discuss Function~\textproc{RL.Select}. If the current state is an end state, this function returns the state representing the default configuration. Otherwise, we use the UCB-V selection policy~\citep{ucbv}, adapted for delayed feedback. Given the state representing the current configuration at time $t$, $c_t$, we choose the action $a_t$ leading to configuration $c_{t+1}$ that maximizes the upper confidence bound:
\begin{align}
    c_{t+1} \triangleq \argmax_c \hat{\mu}_c(t)+\sqrt{2.4 \hat{\sigma}_c^2(t) \frac{\log(v_{c_t})}{v_c}} + \frac{3 b \log(v_{c_t})}{v_c},\label{eq:delayeducbv}
\end{align}
Here, $v_{c_t}$ and $v_c$ are the number of visits to the parent configuration $c_t$ and child configuration $c$ respectively. The average reward obtained till time $t$ after considering delay $\tau$, i.e. $\hat{\mu}_c(t) = \sum_{i=\tau}^t f(c_{i-\tau})\mathbbm{1}(c_{i-\tau}=c)$.
Similarly, $\hat{\sigma}_c^2(t)$ is the empirical variance of reward for configuration $c$ after considering the delay $\tau$. As a practical alternative to the aforementioned estimates, our implementation also supports another estimate of average and variance of reward, following the RAVE (Rapid Action Value Estimation)~\cite{Gelly2007a} approach. This approach shares reward statistics for the same action, invoked in different states, thereby obtaining quality estimates faster. It is known to work well for particularly large search spaces. 

Function~\textproc{RL.Update} updates all of the aforementioned statistics, based on reward values received. More precisely, we update the number of visits to state-action pair $(c_t, a_t)$, present state $c_{t+1}$, and sample mean and variance of accumulated rewards ($\hat{\mu}(c_t, a_t)$ and $\hat{\sigma}^2(c_t, a_t)$). Algorithm~\ref{alg:mcts} shows simplified pseudo-code for Function~\textproc{RL.Optimize}. Given a start state and a search space, this function iterates until a timeout. In each iteration, it selects actions via Function~\textproc{RL.Select} (discussed before), evaluates the performance impact on benchmark $B$, and updates statistics accordingly (using Function~\textproc{RL.Update}). It returns the most promising configuration found until the timeout.




\begin{algorithm}[t!]
\caption{RL: Monte Carlo Tree Search optimization.}
\label{alg:mcts}
\renewcommand{\algorithmiccomment}[1]{// #1}
\begin{small}
\begin{algorithmic}[1]
\State \textbf{Input:} Algorithm $\alg$, configuration space $C$, state $c_0$, benchmark $B$
\State \textbf{Output:} Final parameter configuration
\Function{RL.Optimize}{$\alg, C, c_0$}
\State Initialize $Stat \gets \emptyset$
\For{$t=0, \ldots, \alg.Time$}
\State $\langle c_{t+1}, a_t \rangle\gets$ \Call{RL.Select}{$\alg, C, c_t$}
\State Evaluate the new configuration $r_t \gets$ \Call{B.Evaluate}{$c_{t+1}$}
\State Update $Stat \gets Stat \cup \{\langle c_t, a_t, c_{t+1}, r_t, t\rangle\}$
\State \Call{RL.Update}{$\alg, Stat$}
\EndFor
\State \Return Final parameter configuration $c_T$
\EndFunction
\end{algorithmic}
\end{small}
\end{algorithm}

\section{Theoretical Analysis}
\label{sec:analysis}


We show, under moderately simplifying assumptions, that UDO converges to optimal configurations. UDO uses an extension of HOO algorithm, which provides this type of guarantee (Theorem 6,~\citep{bubeck2011a}). However, we decompose our search space (into a space for heavy and one for light parameters) and delay evaluation feedback (to amortize re-configuration costs). In this section, we sketch out our reasoning for why those changes do not prevent convergence. We provide \revision{proofs for the following theorems online\footnote{\url{https://www.cs.cornell.edu/database/supplementary_proofs.pdf}}}. 

In doing so, we use \textit{expected regret}~\citep{Auer2002} as the metric of convergence. Given a time horizon $T$, expected regret $\mathbb{E}[\reg_T]$ is the sum of differences between the expected performance of the optimal configuration and the configuration achieved by the algorithm at any time step $t \leq T$. If the expected regret of an algorithm grows sublinearly with horizon $T$, it means the algorithm asymptotically converges to optimal configuration as $T \rightarrow \infty$.
\begin{theorem}[Regret of HOO (Theorem 6,~\cite{bubeck2011a})]\label{thm:HOO}
If the performance metric $f$ is smooth around the optimal configuration (Assumption 2 in~\citep{bubeck2011a}) and the upper confidence bounds on performances of all the configurations at depth $h$ create a partition shrinking at the rate $c\rho^h$ with $\rho \in (0,1)$ (Assumption 1 in~\citep{bubeck2011a}), expected regret of HOO is
\begin{align}\label{eq:hoo}
    \mathbb{E}[\reg_T] = O\left(T^{1-\frac{1}{d+2}}(\log T)^{\frac{1}{d+2}}\right)
\end{align}
for a horizon $T > 1$, and $4/c$-near-optimality dimension\footnote{$c$-near-optimality dimension is the smallest $d\geq 0$, such that for all $\varepsilon > 0$, the maximal number of disjoint balls of radius $c\varepsilon$ whose centres can be accommodated in $\mathcal{X}_\varepsilon$ is $O(\varepsilon^{-d})$ (Def. 5 in~\citep{bubeck2011a}). Here, $\varepsilon$-optimal configurations $\mathcal{X}_\varepsilon \triangleq \{ x \in C | f(x) \geq f^* - \varepsilon\}$.
Near-optimality dimension encodes the growth in number of balls needed to pack this set of $\varepsilon$-optimal configurations as $\varepsilon$ increases.} $d$ of $f$.
\end{theorem}
Typically, configuration space $C$ is a bounded subset of $\mathbb{R}^P$ and performance metric $f:C \rightarrow [a,b] \subset \mathbb{R}$. Here, $d$ is of the same order as the number of parameters $P$. HOO uses UCB1~\cite{Auer2002} rather than UCB-V~\citep{ucbv}. For brevity of analysis, we follow the same though the proof technique is similar for any UCB-type (Upper Confidence Bound) algorithm.

\subsection{Regret of Delayed-HOO}
Now, we prove that using delayed-UCB1~\citep{Joulani2013a} instead of UCB1 allows us to propose delayed-HOO and also achieves similar convergence properties.
\begin{theorem}[Regret of Delayed-HOO]\label{thm:delayed-HOO}
Under the same assumptions as Thm.~\ref{thm:HOO}, the expected regret of delayed-HOO is 
\begin{align}\label{eq:delayed_hoo}
    \mathbb{E}[\reg_T] = O\left((1+\tau)T^{1-\frac{1}{d+2}}(\log T)^{\frac{1}{d+2}}\right)
\end{align}
for delay $\tau \geq 0$, horizon $T$, and $4/c$-near-optimality dimension $d$ of $f$. 
\end{theorem}
The bound in Eq.~\eqref{eq:delayed_hoo} is the same as Eq.~\eqref{eq:hoo} with an additional factor $(1+\tau)$, which does not change the convergence in terms of $T$. For delay $\tau = 0$, we retrieve the regret bound of original HOO.

The expected error in estimated expected performance (or reward) of any given configuration at time $T$ is $r(T) = \mathbb{E}[f^* - \hat{f}(c_T)] = \frac{1}{T}\mathbb{E}[\reg_T]$.
Thus, the expected error $\epsilon(T)$ in estimating the expected performance (or reward) of a configuration using delayed-HOO converges at the rate $O\left((1+\tau) \left[\log T/T\right]^{1/(d+2)}\right)$, where $T$ is the number of times the configuration is evaluated.

\subsection{Regret of UDO}
As we have obtained the error bound of the delayed-HOO algorithm, now we can derive bounds for UDO when using delayed-HOO for heavy and light parameters with two different delays.

\begin{theorem}[Regret of UDO]\label{thm:udo}
If we use the delayed-HOO as the delayed-MCTS algorithm with delays $\tau$ and $0$, and time-horizons $T_h$ and $T_l$ for heavy and light parameters respectively, the expected regret of UDO is upper bounded by \begin{align}\label{eq:udo}
    \mathbb{E}[\reg_T] = O\left((1+\tau)T_h^{1-\frac{1}{d+2}}(\mathrm{HOO}^2(T_l) \log T_h)^{\frac{1}{d+2}}\right),
\end{align} 
under the assumptions of Thm.~\ref{thm:HOO}. Here, $\mathrm{HOO}(T_l) \triangleq O\left(\left[\log T_l/T_l\right]^{\frac{1}{d+2}}\right)$.
\end{theorem}

\textit{Deviation in expected performance of the configuration returned by UDO from the optimum is} $O\left((1+\tau)\left[\mathrm{HOO}^2(T_l) \mathrm{HOO}(T_h)\right]^{\frac{1}{d+2}}\right)$. Here, $T_h$ and $T_l$ are the number of steps allotted for the heavy and light parameters respectively. 
Deviation in expected performance of the configuration selected by UDO vanishes as $T_h, T_l \rightarrow \infty$.
%



\definecolor{ddpg}{HTML}{e41a1c}
\definecolor{sarsa}{HTML}{377eb8}
\definecolor{udos}{HTML}{4daf4a}
\definecolor{udo}{HTML}{984ea3}
\definecolor{ottertuneddpg}{HTML}{a65628}
\definecolor{ottertunegp}{HTML}{f781bf}

\section{Experiment Evaluation}
\label{sec:experiments}


After describing our experimental setup (in Section~\ref{sub:expsetup}), we compare UDO against several baselines in Section~\ref{sub:baselines}. Then, in Section~\ref{sub:variants}, we evaluate different UDO variants and justify our primary design decisions. \revision{In Section~\ref{sub:scenarios}, we compare tuning approaches in a diverse range of scenarios.}


\subsection{Experimental Setup}
\label{sub:expsetup}

We consider two standard benchmarks, TPC-C (with 10 warehouses and 32 concurrent requests) and TPC-H (with scaling factor one). We maximize throughput for TPC-C and minimize latency for TPC-H. We automatically tune two popular database management systems, MySQL (version 5.7.29) and Postgres (version 10.15), for maximal performance on those benchmarks. Our parameters include indexing choices (we consider index candidates that are referenced in queries), DBMS configuration parameters, as well as query order in transaction templates (for TPC-C). For TPC-C, \cprotect \revision{we sample configuration quality by running a mix of 4\% \verb|STOCK_LEVEL|, 4\% \verb|DELIVERY|, 4\% \verb|ORDER_STATUS|, 43\% \verb|PAYMENT| and 45\% \verb|NEW_ORDER| transactions for five seconds. Also, we reload a fixed TPC-C snapshot every 10~iterations of UDO's main loop. For TPC-H, we evaluate queries.} We consider 100 tuning parameters for MySQL and 105 parameters for Postgres. The majority of parameters (71) relate to indexing decisions, followed by 19 parameters related to reordering (each parameter represents the position of a query within a transaction template~\cite{Yan2016}), and, finally, parameters representing DBMS configuration parameters (10 parameters for MySQL and 15 for Postgres). For TPC-H, we consider 109 parameters for MySQL and 114 parameters for Postgres (99 of them are related to indexes, the other ones represent DBMS configuration parameters). 


UDO itself is implemented in Python~3, using the OpenAI gym framework. It uses Gurobi (version 9) for cost-based planning. We compare UDO against several baselines that apply RL for universal database optimization without specialized treatment for heavy parameters. Those baselines use out of the box learning algorithms, SARSA~\cite{rummery1994line} and Deep Deterministic Policy Gradient (DDPG)~\cite{Lillicrap2016}, provided by the Keras-RL framework~\cite{kerasrl} for Open AI gym. Prior work on database tuning via reinforcement learning~\cite{Li2018, Zhang2019a} has applied the same framework but to more narrowly defined tuning problems. \revision{We also consider a variant of the latter (using UDO's UCT algorithm without evaluation delays or configuration reordering) that exploits cached configurations. Here, we create database copies for each new heavy parameter configuration encountered and reuse previously created configurations, if available. Our cache uses up to 100 such slots, except for experiments with TPC-H with scaling factor ten where we reduce the number of slots to ten due to higher storage consumption per slot. All baselines discussed so far} can optimize the same search space as UDO. In addition, we compare against combinations of tools that are each targeted at specific tuning problems such as index selection, configuration parameter tuning, or query reordering. Here, we compose solutions proposed for sub-problems by different tools, considering \textit{MySQLTuner}~\cite{mysqltuner}, \textit{PGTuner}~\cite{pgtuner}, and \revision{the Gaussian Process and DDPG++ algorithms~\cite{VanAken2021}, as implemented in the \textit{OtterTune}~\cite{CMUDatabaseGroup2020} tool,} for system parameter tuning, \textit{Quro} for selecting query orders~\cite{Yan2016}, and \textit{Dexter}~\cite{dexter} and \textit{EverSQL}~\cite{eversql} for selecting indexes. \revision{When combining tools, we first optimize transaction code, then parameters, and finally index selections.}

Note that UDO uses no prior training data but optimizes from scratch. Hence, we only consider baselines targeted at the same scenario (i.e., no prior training data). 
Unless noted otherwise, we set UDO's delay $\tau=10$ for the heavy parameter MDP and $b=3$ in UCB-V (Eq.~\eqref{eq:delayeducbv}). We allow up to eight actions (i.e., tuning parameter changes compared to the defaults) per episode for TPC-H and up to 13 for TPC-C (four heavy parameter changes).

All of the following experiments were executed on a server with two Intel Xeon Gold 5218 CPUs with 2.3 GHz (32 physical cores), 384~GB of RAM, and 1~TB of hard disk. 

\subsection{Comparison to Baselines}
\label{sub:baselines}


\pgfplotsset{scaled y ticks=false, log ticks with fixed point}

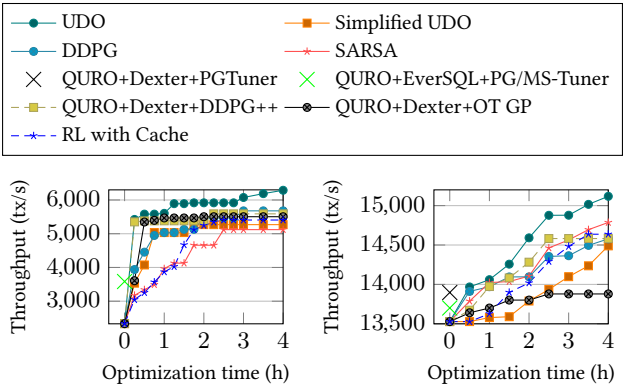
\begin{figure}[t!]
\center
\ref{tpccLegend}
\subfigure[TPC-C performance as a function of optimization time in MySQL.]{
\begin{tikzpicture}
\begin{axis}
[xlabel={Optimization time (h)}, ylabel={Throughput (tx/s)}, width=3.9cm, ylabel near ticks, xlabel near ticks, y label style={font=\small}, x label style={font=\small}, legend entries={UDO, Simplified UDO, DDPG, SARSA, QURO+Dexter+PGTuner, QURO+EverSQL+PG/MS-Tuner, \revision{QURO+Dexter+DDPG++}, \revision{QURO+Dexter+OT GP}, \revision{RL with Cache}}, legend columns=2, legend to name=tpccLegend, legend style={font=\scriptsize}, legend style={cells={align=left,anchor=west}}, legend style={font=\small}, xmax=4, ymin=2330, ymax=6300, ymajorgrids, cycle list name=exotic]
\addplot+[mark size=1.5] table[x index=0, y index=1] {mysql_tpcc.csv};
\addplot+[mark size=1.5] table[x index=0, y index=2] {mysql_tpcc.csv};
\addplot+[mark size=1.5] table[x index=0, y index=3] {mysql_tpcc.csv};
\addplot+[mark size=1.5] table[x index=0, y index=4] {mysql_tpcc.csv};
\addplot+[only marks, mark size=4, mark=x, color=black] coordinates {(0, 0)};
\addplot+[only marks, mark size=4, mark=x, color=green] coordinates {(0, 3583.68)};
\addplot+[mark size=1.5] table[x index=0, y index=5] {mysql_tpcc.csv};
\addplot+[mark size=1.5] table[x index=0, y index=6] {mysql_tpcc.csv};
\addplot+[mark size=1.5] table[x index=0, y index=7] {mysql_tpcc.csv};
\end{axis}
\end{tikzpicture}
\label{fig:mysql_tpcc}
} 
\vspace{-1em}
\subfigure[TPC-C performance as a function of optimization time in Postgres.]{
\begin{tikzpicture}
\begin{axis}[xlabel={Optimization time (h)}, ylabel={Throughput (tx/s)}, width=3.9cm, ylabel near ticks, xlabel near ticks, y label style={font=\small}, x label style={font=\small}, xmax=4, ymin=13500, ymax=15200, ymajorgrids, cycle list name=exotic]
\addplot+[mark size=1.5] table[x index=0, y index=1] {postgres_tpcc.csv};
\addplot+[mark size=1.5] table[x index=0, y index=2] {postgres_tpcc.csv};
\addplot+[mark size=1.5] table[x index=0, y index=3] {postgres_tpcc.csv};
\addplot+[mark size=1.5] table[x index=0, y index=4] {postgres_tpcc.csv};
\addplot+[only marks, mark size=4, mark=x, color=black] coordinates {(0, 13897.07536)};
\addplot+[only marks, mark size=4, mark=x, color=green] coordinates {(0, 13700.76)};
\addplot+[mark size=1.5] table[x index=0, y index=5] {postgres_tpcc.csv};
\addplot+[mark size=1.5] table[x index=0, y index=6] {postgres_tpcc.csv};
\addplot+[mark size=1.5] table[x index=0, y index=7] {postgres_tpcc.csv};
\end{axis}
\end{tikzpicture}
\label{fig:postgres_tpcc}
}
\caption{Comparing UDO to baselines on TPC-C.\label{fig:tpcc_benchmarks}}
\vspace{-1.5em}
\end{figure}



\begin{figure}[t!]
\center
\ref{tpchLegend}
\subfigure[TPC-H performance as a function of optimization time in MySQL.]{
\begin{tikzpicture}
\begin{axis}
[xlabel={Optimization time (h)}, ylabel={Best run time(s)}, width=4cm, ylabel near ticks, xlabel near ticks, y label style={font=\small}, x label style={font=\small}, legend entries={UDO, Simplified UDO, DDPG, SARSA, Dexter+PGTuner, EverSQL+PG/MS-Tuner, \revision{Dexter+OT DDPG++}, \revision{Dexter+OT GP}, \revision{RL with Cache}}, legend columns=2, legend to name=tpchLegend, legend style={font=\scriptsize}, legend style={cells={align=left,anchor=west}}, legend style={font=\small}, xmax=4, ymin=60, ymax=70, ymajorgrids, cycle list name=exotic]
\addplot table[x index=0, y index=1] {mysql_tpch.csv};
\addplot table[x index=0, y index=2] {mysql_tpch.csv};
\addplot table[x index=0, y index=3] {mysql_tpch.csv};
\addplot table[x index=0, y index=4] {mysql_tpch.csv};
\addplot+[only marks, mark size=4, mark=x, color=black] coordinates {(0, 0)};
\addplot+[only marks, mark size=4, mark=x, color=green] coordinates {(0, 69.4)};
\addplot table[x index=0, y index=5] {mysql_tpch.csv};
\addplot table[x index=0, y index=6] {mysql_tpch.csv};
\addplot table[x index=0, y index=7] {mysql_tpch.csv};
\end{axis}
\end{tikzpicture}
\label{fig:mysql_tpch}
}
\vspace{-1em}
\subfigure[TPC-H performance as a function of optimization time in Postgres.]{
\begin{tikzpicture}
\begin{axis}[xlabel={Optimization time (h)}, ylabel={Best run time (s)}, width=4cm, ylabel near ticks, xlabel near ticks, y label style={font=\small}, x label style={font=\small}, xmax=4, ymin=10, ymax=20, ymajorgrids, cycle list name=exotic]
\addplot table[x index=0, y index=1] {postgres_tpch.csv};
\addplot table[x index=0, y index=2] {postgres_tpch.csv};
\addplot table[x index=0, y index=3] {postgres_tpch.csv};
\addplot table[x index=0, y index=4] {postgres_tpch.csv};
\addplot+[only marks, mark size=4, mark=x, color=black] coordinates {(0, 16.881)};
\addplot+[only marks, mark size=4, mark=x, color=green] coordinates {(0, 19.96)};
\addplot table[x index=0, y index=5] {postgres_tpch.csv};
\addplot table[x index=0, y index=6] {postgres_tpch.csv};
\addplot table[x index=0, y index=7] {postgres_tpch.csv};
\end{axis}
\end{tikzpicture}
\label{fig:postgres_tpch}
}
\caption{Comparing UDO to baselines on TPC-H.\label{fig:tpch_benchmarks}}
\vspace{-1.5em}
\end{figure}
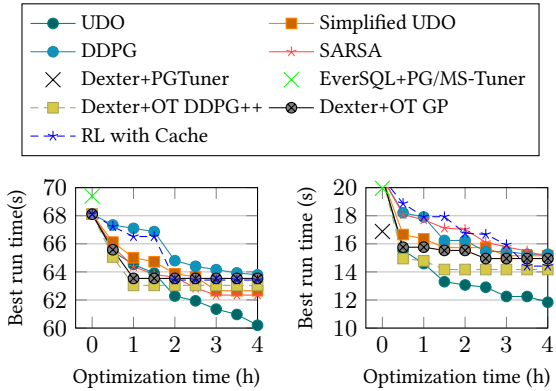

Next, we compare UDO against several baselines. Figure~
\ref{fig:tpcc_benchmarks} reports experimental results for the TPC-C benchmark.  Figure~\ref{fig:tpch_benchmarks} reports results for TPC-H. For TPC-C, we report throughput (of the best configuration found so far) as a function of optimization time. Note that non-iterative baselines are represented as a dot while iterative baselines are represented as a curve. For TPC-H, we report execution time (for TPC-H queries) with the best configuration as a function of optimization time. We optimize for four hours with all baselines. Within the search space defined by our tuning parameters, UDO finds the best configurations for all four combinations of systems and baselines. Generally, the approaches based on reinforcement learning eventually find better solutions than the non-iterative methods. \revision{Among them, UDO performs best, followed in most cases by the combination of DDPG++ (for parameter tuning) and Dexter (for index selection). Configuration caching can sometimes improve over baselines without caches. It is however not as effective as parameter separation and evaluation order optimization, as implemented by UDO.} 

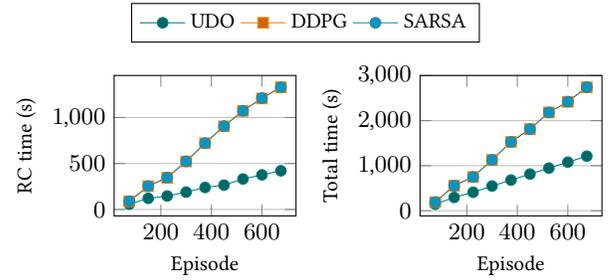
\begin{figure}[t]
\center
\ref{tpccTimeLegendAnalysis}
\subfigure[Reconfiguration time of different RL algorithms for MySQL on TPC-C.]{
\begin{tikzpicture}
\begin{axis}[xlabel={Episode}, ylabel={RC time (s)}, width=4cm, ylabel near ticks, xlabel near ticks, y label style={font=\small}, x label style={font=\small}, ymajorgrids, cycle list name=exotic]
\addplot table[x index=0, y index=3] {mysql_tpcc_time_episode.csv};
\addplot table[x index=0, y index=5] {mysql_tpcc_time_episode.csv};
\addplot table[x index=0, y index=7] {mysql_tpcc_time_episode.csv};
\end{axis}
\end{tikzpicture}
\label{fig:tpcc_time_index}
}
\vspace{-1em}
\subfigure[Total time of different RL algorithms for MySQL on TPC-C.]{
\begin{tikzpicture}
\begin{axis}
[xlabel={Episode}, ylabel={Total time (s)}, width=4cm, ylabel near ticks, xlabel near ticks, y label style={font=\small}, x label style={font=\small}, legend entries={UDO, DDPG, SARSA}, legend columns=4, legend to name=tpccTimeLegendAnalysis, legend style={font=\scriptsize}, legend pos=north east, legend style={font=\small}, ymajorgrids, cycle list name=exotic]
\addplot table[x index=0, y index=4] {mysql_tpcc_time_episode.csv};
\addplot table[x index=0, y index=6] {mysql_tpcc_time_episode.csv};
\addplot table[x index=0, y index=8] {mysql_tpcc_time_episode.csv};
\end{axis}
\end{tikzpicture}
\label{fig:tpcc_time_total}
} 
\caption{Time spent per episode by different RL algorithms when optimizing MySQL for TPC-C.\label{fig:reconftpcc}}
\vspace{-1.5em}
\end{figure}


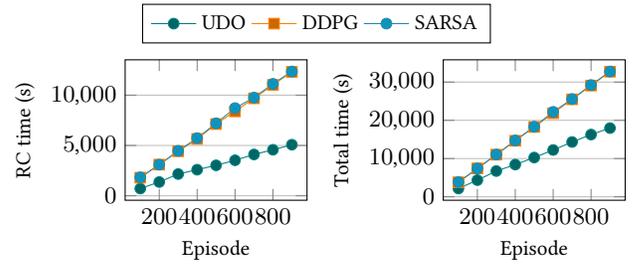
\begin{figure}[t]
\center
\ref{tpchTimeLegendAnalysis}\\
\subfigure[Reconfiguration time of different RL algorithms for Postgres on TPC-H.]{
\begin{tikzpicture}
\begin{axis}[xlabel={Episode}, ylabel={RC time (s)}, width=4cm, ylabel near ticks, xlabel near ticks, y label style={font=\small}, x label style={font=\small}, ymajorgrids, cycle list name=exotic]
\addplot table[x index=0, y index=3] {tpch_time.csv};
\addplot table[x index=0, y index=5] {tpch_time.csv};
\addplot table[x index=0, y index=7] {tpch_time.csv};
\end{axis}
\end{tikzpicture}
\label{fig:tpch_time_index}
}
\vspace{-1em}
\subfigure[Total time of different RL algorithms for Postgres on TPC-H.]{
\begin{tikzpicture}
\begin{axis}
[xlabel={Episode}, ylabel={Total time (s)}, width=4cm, ylabel near ticks, xlabel near ticks, y label style={font=\small}, x label style={font=\small}, legend entries={UDO, DDPG, SARSA}, legend columns=4, legend to name=tpchTimeLegendAnalysis, legend style={font=\scriptsize}, legend pos=north east, legend style={font=\small}, ymajorgrids, cycle list name=exotic]
\addplot table[x index=0, y index=4] {tpch_time.csv};
\addplot table[x index=0, y index=6] {tpch_time.csv};
\addplot table[x index=0, y index=8] {tpch_time.csv};
\end{axis}
\end{tikzpicture}
\label{fig:tpch_time_total}
} 
\caption{Time spent per episode by different RL algorithms when optimizing Postgres for TPC-H.\label{fig:reconftpch}}
\vspace{-1.5em}
\end{figure}

Digging deeper, we analyzed how different RL algorithms spend their time during optimization. Figure~\ref{fig:reconftpcc} shows total (right) and reconfiguration time (left) as a function of the number of episodes for three RL algorithms. Figure~\ref{fig:reconftpch} shows the corresponding numbers for TPC-H. Clearly, UDO iterates faster as it reduces reconfiguration time. For instance, for MySQL running TPC-C, UDO reduces reconfiguration time by a factor of approximately three. This means UDO performs significantly more iterations within the same amount of optimization time, thereby finding promising configurations faster. 

\subsection{Comparison of UDO Variants}
\label{sub:variants}

\begin{figure}[t]
\begin{tikzpicture}
\begin{axis}
[xlabel={Optimization time (h)}, ylabel={Throughput (tx/s)}, width=6cm, height=3.5cm, ylabel near ticks, xlabel near ticks, y label style={font=\small}, x label style={font=\small}, legend entries={Delay=0,Delay=5,Delay=10,Delay=15,Delay=20}, legend columns=1, legend pos=outer north east, legend style={font=\scriptsize}, ymin=5000, ymax=6300, legend style={font=\small}, ymajorgrids, cycle list name=exotic]
\addplot table[x index=0, y index=1] {different_delay.csv};
\addplot table[x index=0, y index=2] {different_delay.csv};
\addplot table[x index=0, y index=3] {different_delay.csv};
\addplot table[x index=0, y index=4] {different_delay.csv};
\addplot table[x index=0, y index=5] {different_delay.csv};
\end{axis}
\end{tikzpicture}
\label{fig:tpch_time_total}
\caption{Impact of delayed feedback on UDO performance (MySQL on TPC-C).\label{fig:delay}}
\vspace{-1.5em}
\end{figure}
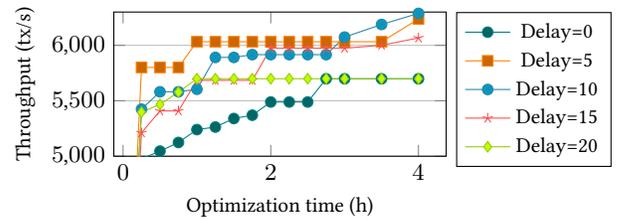

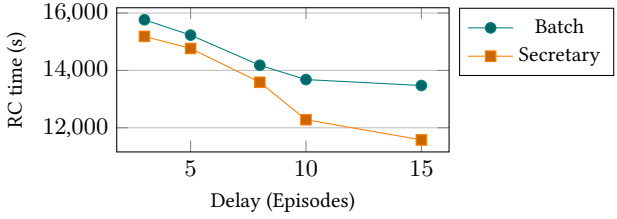
\begin{figure}
\begin{tikzpicture}
\begin{axis}
[xlabel={Delay (Episodes)}, ylabel={RC time (s)}, width=6cm, height=3.5cm, ylabel near ticks, xlabel near ticks, y label style={font=\small}, x label style={font=\small}, legend entries={Batch, Secretary}, legend columns=1, legend style={font=\scriptsize}, legend pos=outer north east, legend style={font=\small}, ymajorgrids, cycle list name=exotic]
\addplot table[x index=0, y index=1] {secretary_approach.csv};
\addplot table[x index=0, y index=2] {secretary_approach.csv};
\end{axis}
\end{tikzpicture}
\label{fig:single_tree_tpcc}
\caption{Impact of evaluation time selection on UDO performance (MySQL on TPC-C).\label{fig:selection}}
\end{figure}

\begin{figure}[t!]
\center
\ref{reorderApproachLegend}\\
\subfigure[Time spent in plan optimization.]{
\begin{tikzpicture}
\begin{axis}[xlabel={Delay (Episodes)}, ylabel={Planning time (s)}, width=3.9cm, ylabel near ticks, xlabel near ticks, y label style={font=\small}, x label style={font=\small}, ymajorgrids, cycle list name=exotic]
\addplot table[x index=0, y index=3] {different_reorder_approach_tpcc.csv};
\addplot table[x index=0, y index=4] {different_reorder_approach_tpcc.csv};
\end{axis}
\end{tikzpicture}
\label{fig:optimization_time_reorder_approach_tpcc}
}
\vspace{-1em}
\subfigure[Time spent in reconfiguration.]{
\begin{tikzpicture}
\begin{axis}
[xlabel={Delay (Episodes)}, ylabel={RC time (s)}, width=3.9cm, ylabel near ticks, xlabel near ticks, y label style={font=\small}, x label style={font=\small}, legend entries={Greedy, Integer Linear Programming}, legend columns=2, legend to name=reorderApproachLegend, legend style={font=\scriptsize}, legend pos=north east, legend style={font=\small}, ymajorgrids, cycle list name=exotic]
\addplot table[x index=0, y index=1] {different_reorder_approach_tpcc.csv};
\addplot table[x index=0, y index=2] {different_reorder_approach_tpcc.csv};
\end{axis}
\end{tikzpicture}
\label{fig:reorder_approach_index_tpcc}
}
\caption{Impact of reconfiguration planning algorithm on UDO performance (MySQL on TPC-C).\label{fig:planner}}
\vspace{-1.5em}
\end{figure}
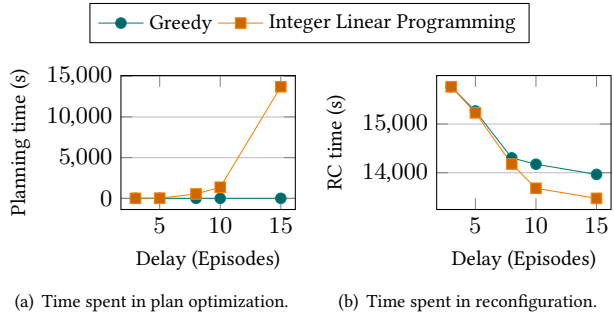

\begin{figure}[t]
\centering
\begin{tikzpicture}
\begin{axis}
[xlabel={Optimization time (h)}, ylabel={Throughput (tx/s)}, width=5.75cm, height=3.5cm, ylabel near ticks, xlabel near ticks, y label style={font=\small}, x label style={font=\small}, legend entries={2-Level UDO, 1-Level UDO, 1-Level UDO D}, legend style={font=\scriptsize}, ymin=4000, ymax=6300, legend pos=outer north east, legend style={font=\small}, ymajorgrids, cycle list name=exotic]
\addplot table[x index=0, y index=1] {single_tree_tpcc.csv};
\addplot table[x index=0, y index=3] {single_tree_tpcc.csv};
\addplot table[x index=0, y index=2] {single_tree_tpcc.csv};
\end{axis}
\end{tikzpicture}
\label{fig:single_tree_tpcc}
\caption{Impact of search space design and search strategy on UDO performance (MySQL on TPC-C).\label{fig:search}}
\end{figure}
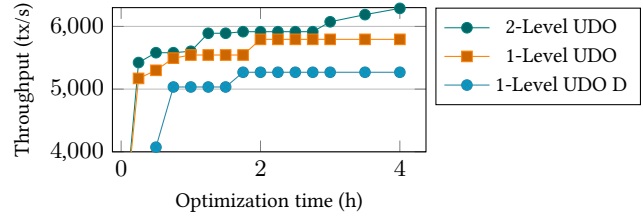


\noindent\textit{Delays.} UDO delays the evaluation of configurations to amortize reconfiguration costs. Of course, there is a trade-off. While delays decrease reconfiguration costs, they may also increase convergence time. Figure~\ref{fig:delay} evaluates UDO with different delay parameters (we measure delay by the maximal number of episodes between request and evaluation). Clearly, disabling delays (Delay=0) leads to slower convergence. Using a delay of five to ten turns out to be the optimal setting (ten is the default setting for our experiments). 

Note that the results in Figure~\ref{fig:delay} relate to the quality of the solution, not to the overheads of optimization. Typically, the rate of improvements slows down as optimization progresses (this effect appears for instance in Figures~\ref{fig:tpcc_benchmarks} and \ref{fig:tpch_benchmarks}). Hence, even small gains in throughput in Figure~\ref{fig:delay} likely translate into significant advantages in terms of optimization time (i.e., optimization time required by weaker approaches to close the gap). 

\noindent\textit{Picking configurations.} Delays are exploited by the evaluation manager to optimize time and order of evaluations. We propose two mechanisms to choose evaluation time. The first one, rather simple, evaluates once the batch of pending evaluation requests reaches a certain size. The second one is more sophisticated and tries to optimize the context in which configurations are evaluated. We compare both methods in Figure~\ref{fig:selection}. Reporting reconfiguration time on the y-axis, we find that ``Secretary-selection'' works best. The gap between the two approaches increases with the delay (a higher delay means more choices in terms of evaluation time). 

\noindent\textit{Ordering configurations.} The second decision made by the evaluation manager relates to the order in which requests, selected for evaluation in a given time slot, are processed. We describe two approaches for request ordering in Section~\ref{sub:ordering}: a simple, greedy algorithm and an approach based on integer linear programming. Figure~\ref{fig:planner} compares the two approaches in terms of optimization time (left) and in terms of reconfiguration time (right), i.e.\ the quality of the generated solution. Clearly, the integer programming approach finds better solutions. The gap increases as the delay (and the number of potential orderings) increases. However, this advantage comes at a steep price. As shown on the left hand side, optimization time increases exponentially and becomes prohibitive for a delay of more than ten (which corresponds to around 100 requests). For our implementation, we switch to the greedy algorithm once delays become prohibitive.

\noindent\textit{1-level vs. 2-level MDP.} Finally, we compare different representations of the search space. Figure~\ref{fig:search} shows corresponding results. Our main version of UDO uses a two-level representation of the search space (separating heavy and light parameter MDPs). In a first step, we remove the separation between the two and apply the same RL algorithm to the 1-Level UDO MDP, introduced in Section~\ref{sec:model}. In a second step, we additionally delay feedback by evaluating configurations only at the end of each episode. Clearly, both of those changes degrade performance, compared to the original version.










\subsection{Scenario Variants}
\label{sub:scenarios}



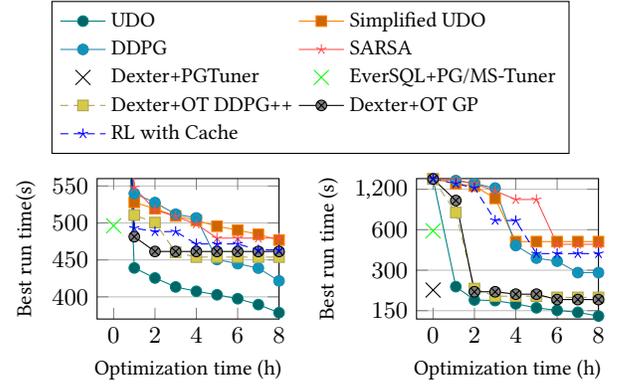
\begin{figure}[t!]
\center
\ref{tpchLegend_sf10}
\subfigure[TPC-H performance as a function of optimization time in MySQL.]{
\begin{tikzpicture}
\begin{axis}[xlabel={Optimization time (h)}, ylabel={Best run time(s)}, width=4cm, ylabel near ticks, xlabel near ticks, y label style={font=\small}, x label style={font=\small}, legend entries={UDO, Simplified UDO, DDPG, SARSA, Dexter+PGTuner, EverSQL+PG/MS-Tuner, \revision{Dexter+OT DDPG++}, \revision{Dexter+OT GP}, \revision{RL with Cache}}, legend columns=2, legend to name=tpchLegend_sf10, legend style={font=\scriptsize}, legend style={cells={align=left,anchor=west}}, legend style={font=\small}, xmax=8, ymin=370, ymax=560, ymajorgrids, cycle list name=exotic]
\addplot table[x index=0, y index=1] {mysql_tpch_sf10.csv};
\addplot table[x index=0, y index=2] {mysql_tpch_sf10.csv};
\addplot table[x index=0, y index=3] {mysql_tpch_sf10.csv};
\addplot table[x index=0, y index=4] {mysql_tpch_sf10.csv};
\addplot+[only marks, mark size=4, mark=x, color=black] coordinates {(0, 0)};
\addplot+[only marks, mark size=4, mark=x, color=green] coordinates {(0, 496.2)};
\addplot table[x index=0, y index=5] {mysql_tpch_sf10.csv};
\addplot table[x index=0, y index=6] {mysql_tpch_sf10.csv};
\addplot table[x index=0, y index=7] {mysql_tpch_sf10.csv};
\end{axis}
\end{tikzpicture}
\label{fig:mysql_tpch_sf10}
}
\vspace{-1em}
\subfigure[TPC-H performance as a function of optimization time in Postgres.]{
\begin{tikzpicture}
\begin{axis}[xlabel={Optimization time (h)}, ylabel={Best run time (s)}, width=4cm, ylabel near ticks, xlabel near ticks, y label style={font=\small}, x label style={font=\small}, xmax=8, ymode=log, ymin=130, ymax=1440, ytick={150,300,600,1200}, ymajorgrids, cycle list name=exotic]
\addplot table[x index=0, y index=1] {postgres_tpch_sf10.csv};
\addplot table[x index=0, y index=2] {postgres_tpch_sf10.csv};
\addplot table[x index=0, y index=3] {postgres_tpch_sf10.csv};
\addplot table[x index=0, y index=4] {postgres_tpch_sf10.csv};
\addplot+[only marks, mark size=4, mark=x, color=black] coordinates {(0, 213.5)};
\addplot+[only marks, mark size=4, mark=x, color=green] coordinates {(0, 589.1)};
\addplot table[x index=0, y index=5] {postgres_tpch_sf10.csv};
\addplot table[x index=0, y index=6] {postgres_tpch_sf10.csv};
\addplot table[x index=0, y index=7] {postgres_tpch_sf10.csv};
\end{axis}
\end{tikzpicture}
\label{fig:postgres_tpch_sf10}
}
\caption{\revision{Comparing UDO to baselines on TPC-H for SF 10.\label{fig:tpch_sf10_benchmarks}}}
\end{figure}

\revision{\noindent\textit{Scaling up.} We increase the scaling factor for TPC-H from one to ten. Figure~\ref{fig:tpch_sf10_benchmarks} reports results for all baselines. The relative tendencies are similar to Figure~\ref{fig:tpch_benchmarks}. However, the spread of run times across different methods is larger. The impact of tuning decisions on performance grows with the data size. For Postgres, at the end of optimization, UDO achieves a 25\% improvement in run time over the second-best baseline (136 versus 181 seconds).}

\input{fig/tpch_space_plus_time}



\begin{figure}[t!]
\center
\ref{tpchIndexTuneLegend}
\subfigure[TPC-H performance as a function of optimization time in MySQL.]{
\begin{tikzpicture}
\begin{axis}
[xlabel={Optimization time (h)}, ylabel={Best run time(s)}, width=4cm, ylabel near ticks, xlabel near ticks, y label style={font=\small}, x label style={font=\small}, legend entries={UDO, DDPG, SARSA, Cost Model, Dexter}, legend columns=3, legend to name=tpchIndexTuneLegend, legend style={font=\scriptsize}, legend style={cells={align=left,anchor=west}}, legend style={font=\small}, xmax=4.5, ymin=800, ymax=1200, ymajorgrids, cycle list name=exotic]
\addplot table[x index=0, y index=1] {mysql_index_tpch_sf10.csv};
\addplot table[x index=0, y index=2] {mysql_index_tpch_sf10.csv};
\addplot table[x index=0, y index=3] {mysql_index_tpch_sf10.csv};
\addplot+[only marks, mark size=4, mark=x, color=green] coordinates {(2.25, 919.568)};
\addplot+[only marks, mark size=4, mark=x, color=red] coordinates {(0, 1030.5)};
\end{axis}
\end{tikzpicture}
\label{fig:mysql_index_tpch}
}
\vspace{-1em}
\subfigure[TPC-H performance as a function of optimization time in Postgres.]{
\begin{tikzpicture}
\begin{axis}[xlabel={Optimization time (h)}, ylabel={Best run time (s)}, width=4cm, ylabel near ticks, xlabel near ticks, y label style={font=\small}, x label style={font=\small}, xmin=0.5, xmax=4.5, ymin=100, ymax=500, ymajorgrids, cycle list name=exotic]
\addplot table[x index=0, y index=1] {postgres_index_tpch_sf10.csv};
\addplot table[x index=0, y index=2] {postgres_index_tpch_sf10.csv};
\addplot table[x index=0, y index=3] {postgres_index_tpch_sf10.csv};
\addplot+[only marks, mark size=4, mark=x, color=green] coordinates {(0.67, 235.5)};
\addplot+[only marks, mark size=4, mark=x, color=red] coordinates {(0.5, 250.2)};
\end{axis}
\end{tikzpicture}
\label{fig:postgres_index_tpch}
}
\caption{\revision{Comparing UDO to baselines for index recommendation (TPC-H SF 10).}\label{fig:tpch_index_tune_benchmarks}}
\end{figure}
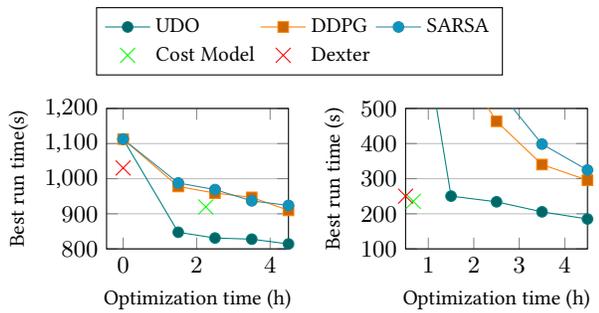

\revision{\noindent\textit{Index recommendation.} UDO is designed to optimize diverse parameters. Nevertheless, we can use it for more narrow problem variants. We evaluate UDO exclusively for index recommendation in Figure~\ref{fig:tpch_index_tune_benchmarks} (using default settings for all database system parameters). We add a new baseline that exploits the query optimizer's cost model: we generate all index candidates and estimate execution costs (via ``explain'' commands) if subsets of indexes are visible to the optimizer. We consider all subsets of up to three index candidates (the number of indexes selected by UDO in the final configuration). While not particularly efficient (the query optimizers of Postgres and MySQL do not directly support what-if analysis), this process identifies the index set that works best according to the optimizer's cost model. UDO ultimately finds better solutions than the baselines. However, the margins are smaller, compared to Figure~\ref{fig:tpch_sf10_benchmarks}. UDO works best for diverse tuning parameters.}




\begin{figure}[t!]
\center
\ref{dynamicLegend}
\subfigure[Varying number of TPC-H query templates used for training.]{
\begin{tikzpicture}
\begin{axis}
[xlabel={\#Templates for training}, ylabel={Best run time(s)}, ylabel near ticks, xlabel near ticks, y label style={font=\small}, x label style={font=\small}, legend entries={UDO, Dexter+OT DDPG++, Dexter+OT GP}, legend columns=3, legend to name=dynamicLegend, legend style={font=\scriptsize}, legend style={cells={align=left,anchor=west}}, legend style={font=\small}, ymajorgrids, width=3.9cm, cycle list name=exotic]
\addplot table[x index=0, y index=1] {tpch_query_change.csv};
\addplot table[x index=0, y index=4] {tpch_query_change.csv};
\addplot table[x index=0, y index=5] {tpch_query_change.csv};
\end{axis}
\end{tikzpicture}
\label{fig:postgresshiftquery}
} 
\vspace{-1em}
\subfigure[Performance for dynamic workload switching every full hour.]{
\begin{tikzpicture}
\begin{axis}
[xlabel={Opt.\ time (h)}, ylabel={Best run time(s)}, width=3.9cm, ylabel near ticks, xlabel near ticks, y label style={font=\small}, x label style={font=\small},ymajorgrids, ymin=50, ymax=250, cycle list name=exotic]
\addplot table[x index=0, y index=1] {tpch_workload_shift.csv};
\addplot table[x index=0, y index=2] {tpch_workload_shift.csv};
\addplot table[x index=0, y index=3] {tpch_workload_shift.csv};
\end{axis}
\end{tikzpicture}
\label{fig:postgresshiftworkload}
}
\caption{\revision{Performance for non-representative training sets and changing workloads (TPC-H SF 10, Postgres).}\label{fig:postgresdynamic}}
\vspace{-1.5em}
\end{figure}
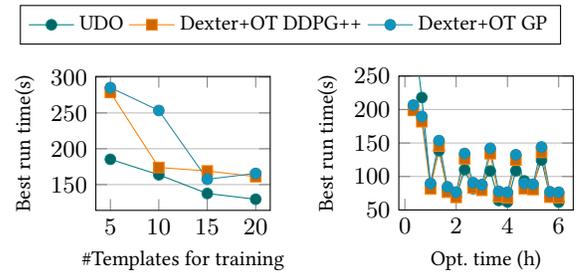


\revision{\noindent\textit{Generalization.} To test generalization, we train UDO and baselines for eight hours on a subset of TPC-H query templates. We show performance of the final configuration for \textit{all queries} in Figure~\ref{fig:postgresshiftquery}. Clearly, training with fewer queries degrades performance on the entire workload. The generalization overheads of UDO are comparable to baselines. E.g., for UDO, performance degrades by about 40\% when considering five instead of 20 templates during training. It is around 70\% for baselines based on Dexter.}


\revision{\noindent\textit{Shifting workload.} In Figure~\ref{fig:postgresshiftworkload}, we report results for a dynamic workload. We switch back between TPC-H query templates with odd numbers (i.e., Q1, Q3, etc.) and templates with even numbers every hour. Figure~\ref{fig:postgresshiftworkload} reports run time for the current half of queries as a function of optimization time. For DDPG++ and OT GP, we use indexes proposed by Dexter for each of the two workload parts. As the indexes proposed by Dexter lead to one problematic query running for more than one hour, we added one more index from the final configuration generated by UDO for the baseline (index on the ``L\_PARTKEY'' column of the ``Lineitem'' table). The presented results therefore correspond to upper bounds on performance for all approaches except for UDO. We see spikes for all baselines, whenever the workload changes. The magnitude of the spikes decreases over time, showing that all approaches converge to a configuration that compromises between the two workload parts. Considering aggregate run times for both workload parts, UDO still performs about 5\% better than the nearest baseline.}

\section{Related Work}
\label{sec:related}

Recently, there has been significant interest in using machine learning for database tuning~\cite{Woltmann,Park2020,Hilprecht2019a,Ma2020,skinnerDB}. Our work falls into the same, broad category as it exploits RL. Prior work typically focuses on specific tuning choices such as system configuration parameters~\cite{Li2018,Zhang2019a,Zhang1910}, index selection~\cite{Sharma2018, Sadri2020}, or data partitioning~\cite{Hilprecht2020, Yang2020}. \revision{Our goal is to support a broad set of tuning choices} via one unified approach. We show, via our experiments, that technical ideas such as parameter partitioning and reconfiguration planning are beneficial in this context.


Traditionally, tuning decisions in a database system are made based on simplifying execution cost models. This often leads to sub-optimal choices in practice~\cite{Borovica2012, Gubichev2015}. UDO does not use any simplifying cost model. Instead, it exclusively uses feedback obtained via trial runs to identify promising configurations. In that, it also differs from a significant fraction of prior work using machine learning for database tuning~\cite{Zhang1910, kipf2018}. Many corresponding approaches rely on a-priori training data, obtained from representative workloads. UDO assumes no prior training data and learns (near-)optimal configurations from scratch. This makes optimization expensive (in the order of hours for our experiments) but avoids generalization errors and the need for training data. UDO will be demonstrated at the upcoming SIGMOD'21 conference~\cite{Wang2021}.




\section{Conclusion}
\label{sec:conclusion}

We presented a system, UDO, for optimizing various tuning parameters by a unified approach. Our experiments show \revision{that parameter separation and delayed learning yield significant improvements.}



\bibliographystyle{ACM-Reference-Format}
\bibliography{references}


\begin{thebibliography}{00}


\ifx \showCODEN    \undefined \def \showCODEN     #1{\unskip}     \fi
\ifx \showDOI      \undefined \def \showDOI       #1{#1}\fi
\ifx \showISBNx    \undefined \def \showISBNx     #1{\unskip}     \fi
\ifx \showISBNxiii \undefined \def \showISBNxiii  #1{\unskip}     \fi
\ifx \showISSN     \undefined \def \showISSN      #1{\unskip}     \fi
\ifx \showLCCN     \undefined \def \showLCCN      #1{\unskip}     \fi
\ifx \shownote     \undefined \def \shownote      #1{#1}          \fi
\ifx \showarticletitle \undefined \def \showarticletitle #1{#1}   \fi
\ifx \showURL      \undefined \def \showURL       {\relax}        \fi
\providecommand\bibfield[2]{#2}
\providecommand\bibinfo[2]{#2}
\providecommand\natexlab[1]{#1}
\providecommand\showeprint[2][]{arXiv:#2}

\bibitem[\protect\citeauthoryear{??}{pgt}{2020}]%
        {pgtuner}
 \bibinfo{year}{2020}\natexlab{}.
\newblock \bibinfo{title}{https://github.com/jfcoz/postgresqltuner}.
\newblock   (\bibinfo{year}{2020}).
\newblock


\bibitem[\protect\citeauthoryear{??}{ker}{2020}]%
        {kerasrl}
 \bibinfo{year}{2020}\natexlab{}.
\newblock \bibinfo{title}{https://github.com/keras-rl/keras-rl}.
\newblock   (\bibinfo{year}{2020}).
\newblock


\bibitem[\protect\citeauthoryear{??}{mys}{2020}]%
        {mysqltuner}
 \bibinfo{year}{2020}\natexlab{}.
\newblock \bibinfo{title}{{https://github.com/major/MySQLTuner-perl}}.
\newblock   (\bibinfo{year}{2020}).
\newblock


\bibitem[\protect\citeauthoryear{??}{dex}{2021}]%
        {dexter}
 \bibinfo{year}{2021}\natexlab{}.
\newblock \bibinfo{title}{https://github.com/ankane/dexter}.
\newblock   (\bibinfo{year}{2021}).
\newblock


\bibitem[\protect\citeauthoryear{??}{eve}{2021}]%
        {eversql}
 \bibinfo{year}{2021}\natexlab{}.
\newblock \bibinfo{title}{https://www.eversql.com/}.
\newblock   (\bibinfo{year}{2021}).
\newblock


\bibitem[\protect\citeauthoryear{Audibert, Munos, and Szepesv{\'a}ri}{Audibert
  et~al\mbox{.}}{2007}]%
        {ucbv}
\bibfield{author}{\bibinfo{person}{Jean-Yves Audibert},
  \bibinfo{person}{R{\'e}mi Munos}, {and} \bibinfo{person}{Csaba
  Szepesv{\'a}ri}.} \bibinfo{year}{2007}\natexlab{}.
\newblock \showarticletitle{Tuning Bandit Algorithms in Stochastic
  Environments}. In \bibinfo{booktitle}{{\em Algorithmic Learning Theory}},
  \bibfield{editor}{\bibinfo{person}{Marcus Hutter}, \bibinfo{person}{Rocco~A.
  Servedio}, {and} \bibinfo{person}{Eiji Takimoto}} (Eds.).
  \bibinfo{publisher}{Springer Berlin Heidelberg}, \bibinfo{address}{Berlin,
  Heidelberg}, \bibinfo{pages}{150--165}.
\newblock


\bibitem[\protect\citeauthoryear{Auer, Cesa-bianchi, and Fischer}{Auer
  et~al\mbox{.}}{2002}]%
        {Auer2002}
\bibfield{author}{\bibinfo{person}{P Auer}, \bibinfo{person}{N Cesa-bianchi},
  {and} \bibinfo{person}{P Fischer}.} \bibinfo{year}{2002}\natexlab{}.
\newblock \showarticletitle{{Finite time analysis of the multiarmed bandit
  problem}}.
\newblock \bibinfo{journal}{{\em Machine Learning\/}} \bibinfo{volume}{47},
  \bibinfo{number}{2-3} (\bibinfo{year}{2002}), \bibinfo{pages}{235--256}.
\newblock


\bibitem[\protect\citeauthoryear{Borovica, Alagiannis, and Ailamaki}{Borovica
  et~al\mbox{.}}{2012}]%
        {Borovica2012}
\bibfield{author}{\bibinfo{person}{Renata Borovica}, \bibinfo{person}{Ioannis
  Alagiannis}, {and} \bibinfo{person}{Anastasia Ailamaki}.}
  \bibinfo{year}{2012}\natexlab{}.
\newblock \showarticletitle{{Automated physical designers: what you see is
  (not) what you get}}. In \bibinfo{booktitle}{{\em Proceedings of the Fifth
  International Workshop on Testing Database Systems}}.
  \bibinfo{pages}{9:1----9:6}.
\newblock
\showISBNx{978-1-4503-1429-9}
\showDOI{%
\url{https://doi.org/10.1145/2304510.2304522}}


\bibitem[\protect\citeauthoryear{Bubeck, Munos, Stoltz, and
  Szepesv{\'a}ri}{Bubeck et~al\mbox{.}}{2011}]%
        {bubeck2011a}
\bibfield{author}{\bibinfo{person}{S{\'e}bastien Bubeck},
  \bibinfo{person}{R{\'e}mi Munos}, \bibinfo{person}{Gilles Stoltz}, {and}
  \bibinfo{person}{Csaba Szepesv{\'a}ri}.} \bibinfo{year}{2011}\natexlab{}.
\newblock \showarticletitle{X-Armed Bandits.}
\newblock \bibinfo{journal}{{\em Journal of Machine Learning Research\/}}
  \bibinfo{volume}{12}, \bibinfo{number}{5} (\bibinfo{year}{2011}).
\newblock


\bibitem[\protect\citeauthoryear{Chaudhuri}{Chaudhuri}{2004}]%
        {Chaudhuri2004}
\bibfield{author}{\bibinfo{person}{Surajit Chaudhuri}.}
  \bibinfo{year}{2004}\natexlab{}.
\newblock \showarticletitle{{Index selection for databases: A hardness study
  and a principled heuristic solution}}.
\newblock \bibinfo{journal}{{\em KDE\/}} \bibinfo{volume}{16},
  \bibinfo{number}{11} (\bibinfo{year}{2004}), \bibinfo{pages}{1313--1323}.
\newblock
\showURL{%
\url{http://ieeexplore.ieee.org/xpls/abs}}


\bibitem[\protect\citeauthoryear{Chaudhuri, Narasayya, and Ramamurty}{Chaudhuri
  et~al\mbox{.}}{2009}]%
        {Chaudhuri2009b}
\bibfield{author}{\bibinfo{person}{Surajit Chaudhuri}, \bibinfo{person}{V
  Narasayya}, {and} \bibinfo{person}{Ravi Ramamurty}.}
  \bibinfo{year}{2009}\natexlab{}.
\newblock \showarticletitle{{Exact cardinality query optimization for optimizer
  testing}}. In \bibinfo{booktitle}{{\em VLDB}}. \bibinfo{pages}{994--1005}.
\newblock
\showISBNx{9781605589480}
\showISSN{2150-8097}
\showDOI{%
\url{https://doi.org/10.14778/1687627.1687739}}


\bibitem[\protect\citeauthoryear{{CMU Database Group}}{{CMU Database
  Group}}{2020}]%
        {CMUDatabaseGroup2020}
\bibfield{author}{\bibinfo{person}{{CMU Database Group}}.}
  \bibinfo{year}{2020}\natexlab{}.
\newblock \bibinfo{title}{https://github.com/cmu-db/ottertune}.
\newblock   (\bibinfo{year}{2020}).
\newblock


\bibitem[\protect\citeauthoryear{Coquelin and Munos}{Coquelin and
  Munos}{2007}]%
        {Coquelin2007a}
\bibfield{author}{\bibinfo{person}{Pierre-Arnaud Coquelin} {and}
  \bibinfo{person}{R{\'{e}}mi Munos}.} \bibinfo{year}{2007}\natexlab{}.
\newblock \showarticletitle{{Bandit Algorithms for Tree Search}}.
\newblock \bibinfo{journal}{{\em Arxiv preprint cs0703062\/}}
  \bibinfo{volume}{23}, \bibinfo{number}{March} (\bibinfo{year}{2007}),
  \bibinfo{pages}{67--74}.
\newblock
\showISBNx{0-9749039-3-0}
\showeprint[arxiv]{arXiv:cs/0703062v1}
\showURL{%
\url{http://arxiv.org/abs/cs/0703062}}


\bibitem[\protect\citeauthoryear{Ding, Das, Marcus, Wu, Chaudhuri, and
  Narasayya}{Ding et~al\mbox{.}}{2019}]%
        {Ding2019}
\bibfield{author}{\bibinfo{person}{Bailu Ding}, \bibinfo{person}{Sudipto Das},
  \bibinfo{person}{Ryan Marcus}, \bibinfo{person}{Wentao Wu},
  \bibinfo{person}{Surajit Chaudhuri}, {and} \bibinfo{person}{Vivek~R.
  Narasayya}.} \bibinfo{year}{2019}\natexlab{}.
\newblock \showarticletitle{{AI meets AI: Leveraging query executions to
  improve index recommendations}}. In \bibinfo{booktitle}{{\em SIGMOD}}.
  \bibinfo{pages}{1241--1258}.
\newblock
\showISBNx{9781450356435}
\showISSN{07308078}
\showDOI{%
\url{https://doi.org/10.1145/3299869.3324957}}


\bibitem[\protect\citeauthoryear{Gelly and Silver}{Gelly and Silver}{2007}]%
        {Gelly2007a}
\bibfield{author}{\bibinfo{person}{Sylvain Gelly} {and} \bibinfo{person}{David
  Silver}.} \bibinfo{year}{2007}\natexlab{}.
\newblock \showarticletitle{{Combining online and offline knowledge in UCT}}.
\newblock \bibinfo{journal}{{\em Proceedings of the 24th international
  conference on Machine learning - ICML '07\/}} (\bibinfo{year}{2007}),
  \bibinfo{pages}{273--280}.
\newblock
\showISBNx{9781595937933}
\showDOI{%
\url{https://doi.org/10.1145/1273496.1273531}}


\bibitem[\protect\citeauthoryear{Gubichev, Boncz, Kemper, and Neumann}{Gubichev
  et~al\mbox{.}}{2015}]%
        {Gubichev2015}
\bibfield{author}{\bibinfo{person}{Andrey Gubichev}, \bibinfo{person}{Peter
  Boncz}, \bibinfo{person}{Alfons Kemper}, {and} \bibinfo{person}{Thomas
  Neumann}.} \bibinfo{year}{2015}\natexlab{}.
\newblock \showarticletitle{{How good are query optimizers, really?}}
\newblock \bibinfo{journal}{{\em PVLDB\/}} \bibinfo{volume}{9},
  \bibinfo{number}{3} (\bibinfo{year}{2015}), \bibinfo{pages}{204--215}.
\newblock


\bibitem[\protect\citeauthoryear{Hill}{Hill}{2009}]%
        {Hill2009}
\bibfield{author}{\bibinfo{person}{Theodore~P. Hill}.}
  \bibinfo{year}{2009}\natexlab{}.
\newblock \showarticletitle{{Knowing when to stop}}.
\newblock \bibinfo{journal}{{\em American Scientist\/}} \bibinfo{volume}{97},
  \bibinfo{number}{2} (\bibinfo{year}{2009}), \bibinfo{pages}{126--133}.
\newblock
\showISSN{00030996}
\showDOI{%
\url{https://doi.org/10.1511/2009.77.126}}


\bibitem[\protect\citeauthoryear{Hilprecht, Binnig, and R{\"{o}}hm}{Hilprecht
  et~al\mbox{.}}{2019}]%
        {Hilprecht2019a}
\bibfield{author}{\bibinfo{person}{Benjamin Hilprecht},
  \bibinfo{person}{Carsten Binnig}, {and} \bibinfo{person}{Uwe R{\"{o}}hm}.}
  \bibinfo{year}{2019}\natexlab{}.
\newblock \showarticletitle{{Towards learning a partitioning advisor with deep
  reinforcement learning}}.
\newblock \bibinfo{journal}{{\em SIGMOD\/}} (\bibinfo{year}{2019}).
\newblock
\showISBNx{9781450368025}
\showISSN{07308078}
\showDOI{%
\url{https://doi.org/10.1145/3329859.3329876}}
\showeprint[arxiv]{arXiv:1904.01279v1}


\bibitem[\protect\citeauthoryear{Hilprecht, Binnig, and R{\"{o}}hm}{Hilprecht
  et~al\mbox{.}}{2020}]%
        {Hilprecht2020}
\bibfield{author}{\bibinfo{person}{Benjamin Hilprecht},
  \bibinfo{person}{Carsten Binnig}, {and} \bibinfo{person}{Uwe R{\"{o}}hm}.}
  \bibinfo{year}{2020}\natexlab{}.
\newblock \showarticletitle{{Learning a Partitioning Advisor for Cloud
  Databases}}.
\newblock \bibinfo{journal}{{\em Proceedings of the ACM SIGMOD International
  Conference on Management of Data\/}} (\bibinfo{year}{2020}),
  \bibinfo{pages}{143--157}.
\newblock
\showISBNx{9781450367356}
\showISSN{07308078}
\showDOI{%
\url{https://doi.org/10.1145/3318464.3389704}}


\bibitem[\protect\citeauthoryear{Joulani, Gy{\"{o}}rgy, and Szepesvari}{Joulani
  et~al\mbox{.}}{2013}]%
        {Joulani2013a}
\bibfield{author}{\bibinfo{person}{Pooria Joulani},
  \bibinfo{person}{Andr{\'{a}}s Gy{\"{o}}rgy}, {and} \bibinfo{person}{Csaba
  Szepesvari}.} \bibinfo{year}{2013}\natexlab{}.
\newblock \showarticletitle{{Online learning under delayed feedback}}.
\newblock \bibinfo{journal}{{\em 30th International Conference on Machine
  Learning, ICML 2013\/}} \bibinfo{number}{PART 3} (\bibinfo{year}{2013}),
  \bibinfo{pages}{2503--2511}.
\newblock
\showDOI{%
\url{https://doi.org/10.14288/1.0044651}}
\showeprint[arxiv]{1306.0686}


\bibitem[\protect\citeauthoryear{Kipf, Kipf, Radke, Leis, Boncz, and
  Kemper}{Kipf et~al\mbox{.}}{2018}]%
        {kipf2018}
\bibfield{author}{\bibinfo{person}{Andreas Kipf}, \bibinfo{person}{Thomas
  Kipf}, \bibinfo{person}{Bernhard Radke}, \bibinfo{person}{Viktor Leis},
  \bibinfo{person}{Peter Boncz}, {and} \bibinfo{person}{Alfons Kemper}.}
  \bibinfo{year}{2018}\natexlab{}.
\newblock \showarticletitle{{Learned cardinalities: estimating correlated joins
  with deep learning}}. In \bibinfo{booktitle}{{\em CIDR}}.
\newblock
\showeprint[arxiv]{1809.00677}
\showURL{%
\url{http://arxiv.org/abs/1809.00677}}


\bibitem[\protect\citeauthoryear{Kocsis and Szepesv{\'{a}}ri}{Kocsis and
  Szepesv{\'{a}}ri}{2006}]%
        {Kocsis2006}
\bibfield{author}{\bibinfo{person}{Levente Kocsis} {and} \bibinfo{person}{C
  Szepesv{\'{a}}ri}.} \bibinfo{year}{2006}\natexlab{}.
\newblock \showarticletitle{{Bandit based monte-carlo planning}}. In
  \bibinfo{booktitle}{{\em European Conf. on Machine Learning}}.
  \bibinfo{pages}{282--293}.
\newblock
\showURL{%
\url{http://www.springerlink.com/index/D232253353517276.pdf}}


\bibitem[\protect\citeauthoryear{Li, Zhou, Li, and Gao}{Li
  et~al\mbox{.}}{2018}]%
        {Li2018}
\bibfield{author}{\bibinfo{person}{Guoliang Li}, \bibinfo{person}{Xuanhe Zhou},
  \bibinfo{person}{Shifu Li}, {and} \bibinfo{person}{Bo Gao}.}
  \bibinfo{year}{2018}\natexlab{}.
\newblock \showarticletitle{{QTune: A QueryAware database tuning system with
  deep reinforcement learning}}.
\newblock \bibinfo{journal}{{\em PVLDB\/}} \bibinfo{volume}{12},
  \bibinfo{number}{12} (\bibinfo{year}{2018}), \bibinfo{pages}{2118--2130}.
\newblock
\showISSN{21508097}
\showDOI{%
\url{https://doi.org/10.14778/3352063.3352129}}


\bibitem[\protect\citeauthoryear{Lillicrap, Hunt, Pritzel, Heess, Erez, Tassa,
  Silver, and Wierstra}{Lillicrap et~al\mbox{.}}{2016}]%
        {Lillicrap2016}
\bibfield{author}{\bibinfo{person}{Timothy~P. Lillicrap},
  \bibinfo{person}{Jonathan~J. Hunt}, \bibinfo{person}{Alexander Pritzel},
  \bibinfo{person}{Nicolas Heess}, \bibinfo{person}{Tom Erez},
  \bibinfo{person}{Yuval Tassa}, \bibinfo{person}{David Silver}, {and}
  \bibinfo{person}{Daan Wierstra}.} \bibinfo{year}{2016}\natexlab{}.
\newblock \showarticletitle{{Continuous control with deep reinforcement
  learning}}.
\newblock \bibinfo{journal}{{\em 4th International Conference on Learning
  Representations, ICLR 2016 - Conference Track Proceedings\/}}
  (\bibinfo{year}{2016}).
\newblock
\showeprint[arxiv]{1509.02971}


\bibitem[\protect\citeauthoryear{Ma, Ding, Das, and Swaminathan}{Ma
  et~al\mbox{.}}{2020}]%
        {Ma2020}
\bibfield{author}{\bibinfo{person}{Lin Ma}, \bibinfo{person}{Bailu Ding},
  \bibinfo{person}{Sudipto Das}, {and} \bibinfo{person}{Adith Swaminathan}.}
  \bibinfo{year}{2020}\natexlab{}.
\newblock \showarticletitle{{Active Learning for ML Enhanced Database
  Systems}}. In \bibinfo{booktitle}{{\em SIGMOD}}. \bibinfo{pages}{175--191}.
\newblock
\showISBNx{9781450367356}
\showISSN{07308078}
\showDOI{%
\url{https://doi.org/10.1145/3318464.3389768}}


\bibitem[\protect\citeauthoryear{Park, Zhong, and Mozafari}{Park
  et~al\mbox{.}}{2020}]%
        {Park2020}
\bibfield{author}{\bibinfo{person}{Yongjoo Park}, \bibinfo{person}{Shucheng
  Zhong}, {and} \bibinfo{person}{Barzan Mozafari}.}
  \bibinfo{year}{2020}\natexlab{}.
\newblock \showarticletitle{{QuickSel: Quick Selectivity Learning with Mixture
  Models}}. In \bibinfo{booktitle}{{\em SIGMOD}}. \bibinfo{pages}{1017--1033}.
\newblock
\showISBNx{9781450367356}
\showISSN{07308078}
\showDOI{%
\url{https://doi.org/10.1145/3318464.3389727}}
\showeprint[arxiv]{1812.10568}


\bibitem[\protect\citeauthoryear{Rummery and Niranjan}{Rummery and
  Niranjan}{1994}]%
        {rummery1994line}
\bibfield{author}{\bibinfo{person}{Gavin~A Rummery} {and}
  \bibinfo{person}{Mahesan Niranjan}.} \bibinfo{year}{1994}\natexlab{}.
\newblock \bibinfo{booktitle}{{\em On-line Q-learning using connectionist
  systems}}. Vol.~\bibinfo{volume}{37}.
\newblock \bibinfo{publisher}{University of Cambridge, Department of
  Engineering Cambridge, UK}.
\newblock


\bibitem[\protect\citeauthoryear{Sadri, Gruenwald, and Lead}{Sadri
  et~al\mbox{.}}{2020}]%
        {Sadri2020}
\bibfield{author}{\bibinfo{person}{Zahra Sadri}, \bibinfo{person}{Le
  Gruenwald}, {and} \bibinfo{person}{Eleazar Lead}.}
  \bibinfo{year}{2020}\natexlab{}.
\newblock \showarticletitle{{DRLindex: Deep reinforcement learning index
  advisor for a cluster database}}.
\newblock \bibinfo{journal}{{\em ACM International Conference Proceeding
  Series\/}} (\bibinfo{year}{2020}).
\newblock
\showISBNx{9781450375030}
\showDOI{%
\url{https://doi.org/10.1145/3410566.3410603}}


\bibitem[\protect\citeauthoryear{Sharma, Schuhknecht, and Dittrich}{Sharma
  et~al\mbox{.}}{2018}]%
        {Sharma2018}
\bibfield{author}{\bibinfo{person}{Ankur Sharma}, \bibinfo{person}{Felix~Martin
  Schuhknecht}, {and} \bibinfo{person}{Jens Dittrich}.}
  \bibinfo{year}{2018}\natexlab{}.
\newblock \showarticletitle{{The case for automatic database administration
  using deep reinforcement learning}}.
\newblock \bibinfo{journal}{{\em arXiv\/}} (\bibinfo{year}{2018}),
  \bibinfo{pages}{1--9}.
\newblock
\showISSN{23318422}
\showeprint[arxiv]{1801.05643}


\bibitem[\protect\citeauthoryear{Trummer}{Trummer}{2019}]%
        {Trummer2019}
\bibfield{author}{\bibinfo{person}{Immanuel Trummer}.}
  \bibinfo{year}{2019}\natexlab{}.
\newblock \showarticletitle{{Exact cardinality query optimization with bounded
  execution cost}}. In \bibinfo{booktitle}{{\em SIGMOD}}.
  \bibinfo{pages}{2--17}.
\newblock


\bibitem[\protect\citeauthoryear{Trummer, Wang, Maram, Moseley, Jo, and
  Antonakakis}{Trummer et~al\mbox{.}}{2019}]%
        {skinnerDB}
\bibfield{author}{\bibinfo{person}{Immanuel Trummer}, \bibinfo{person}{Junxiong
  Wang}, \bibinfo{person}{Deepak Maram}, \bibinfo{person}{Samuel Moseley},
  \bibinfo{person}{Saehan Jo}, {and} \bibinfo{person}{Joseph Antonakakis}.}
  \bibinfo{year}{2019}\natexlab{}.
\newblock \showarticletitle{{SkinnerDB: regret-bounded query evaluation via
  reinforcement learning}}. In \bibinfo{booktitle}{{\em SIGMOD}}.
  \bibinfo{pages}{1039--1050}.
\newblock


\bibitem[\protect\citeauthoryear{{Van Aken}, Yang, Brillard, Fiorino, Zhang,
  Bilien, and Pavlo}{{Van Aken} et~al\mbox{.}}{2021}]%
        {VanAken2021}
\bibfield{author}{\bibinfo{person}{Dana {Van Aken}}, \bibinfo{person}{Dongsheng
  Yang}, \bibinfo{person}{Sebastien Brillard}, \bibinfo{person}{Ari Fiorino},
  \bibinfo{person}{Bohan Zhang}, \bibinfo{person}{Christian Bilien}, {and}
  \bibinfo{person}{Andrew Pavlo}.} \bibinfo{year}{2021}\natexlab{}.
\newblock \showarticletitle{{An inquiry into machine learning-based automatic
  configuration tuning services on real-world database management systems}}.
\newblock \bibinfo{journal}{{\em Proceedings of the VLDB Endowment\/}}
  \bibinfo{volume}{14}, \bibinfo{number}{7} (\bibinfo{year}{2021}),
  \bibinfo{pages}{1241--1253}.
\newblock
\showISSN{21508097}
\showDOI{%
\url{https://doi.org/10.14778/3450980.3450992}}


\bibitem[\protect\citeauthoryear{Wang, Trummer, and Basu}{Wang
  et~al\mbox{.}}{2021}]%
        {Wang2021}
\bibfield{author}{\bibinfo{person}{Junxiong Wang}, \bibinfo{person}{Immanuel
  Trummer}, {and} \bibinfo{person}{Debabrota Basu}.}
  \bibinfo{year}{2021}\natexlab{}.
\newblock \showarticletitle{{Demonstrating UDO: A Unified Approach for
  OptimizingTransaction Code, Physical Design, and System Parameters via
  Reinforcement Learning}}. In \bibinfo{booktitle}{{\em SIGMOD}}.
\newblock


\bibitem[\protect\citeauthoryear{Woltmann, Hartmann, Thiele, and
  Habich}{Woltmann et~al\mbox{.}}{2019}]%
        {Woltmann}
\bibfield{author}{\bibinfo{person}{Lucas Woltmann}, \bibinfo{person}{Claudio
  Hartmann}, \bibinfo{person}{Maik Thiele}, {and} \bibinfo{person}{Dirk
  Habich}.} \bibinfo{year}{2019}\natexlab{}.
\newblock \showarticletitle{{Cardinality estimation with local deep learning
  models}}. In \bibinfo{booktitle}{{\em aiDM}}.
\newblock
\showISBNx{9781450368025}


\bibitem[\protect\citeauthoryear{Yan and Cheung}{Yan and Cheung}{2016}]%
        {Yan2016}
\bibfield{author}{\bibinfo{person}{Cong Yan} {and} \bibinfo{person}{Alvin
  Cheung}.} \bibinfo{year}{2016}\natexlab{}.
\newblock \showarticletitle{{Leveraging Lock Contention to Improve OLTP
  Application Performance}}. In \bibinfo{booktitle}{{\em VLDBJ}},
  Vol.~\bibinfo{volume}{9}. \bibinfo{pages}{444--455}.
\newblock
\showISSN{21508097}
\showDOI{%
\url{https://doi.org/10.14778/2876473.2876479}}


\bibitem[\protect\citeauthoryear{Yang, Chandramouli, Wang, Gehrke, Li, Minhas,
  Larson, Kossmann, and Acharya}{Yang et~al\mbox{.}}{2020}]%
        {Yang2020}
\bibfield{author}{\bibinfo{person}{Zongheng Yang}, \bibinfo{person}{Badrish
  Chandramouli}, \bibinfo{person}{Chi Wang}, \bibinfo{person}{Johannes Gehrke},
  \bibinfo{person}{Yinan Li}, \bibinfo{person}{Umar~Farooq Minhas},
  \bibinfo{person}{Per~{\AA}ke Larson}, \bibinfo{person}{Donald Kossmann},
  {and} \bibinfo{person}{Rajeev Acharya}.} \bibinfo{year}{2020}\natexlab{}.
\newblock \showarticletitle{{Qd-tree: Learning data layouts for big data
  analytics}}.
\newblock \bibinfo{journal}{{\em arXiv\/}} \bibinfo{number}{2}
  (\bibinfo{year}{2020}), \bibinfo{pages}{193--208}.
\newblock
\showISBNx{9781450367356}
\showISSN{23318422}


\bibitem[\protect\citeauthoryear{Zhang, Aken, Wang, Dai, Jiang, Lao, Sheng,
  Pavlo, and Gordon}{Zhang et~al\mbox{.}}{1910}]%
        {Zhang1910}
\bibfield{author}{\bibinfo{person}{Bohan Zhang}, \bibinfo{person}{Dana~Van
  Aken}, \bibinfo{person}{Justin Wang}, \bibinfo{person}{Tao Dai},
  \bibinfo{person}{Shuli Jiang}, \bibinfo{person}{Jacky Lao},
  \bibinfo{person}{Siyuan Sheng}, \bibinfo{person}{Andrew Pavlo}, {and}
  \bibinfo{person}{Geoffrey~J Gordon}.} \bibinfo{year}{1910}\natexlab{}.
\newblock \showarticletitle{{A demonstration of the OtterTune automatic
  database management system tuning service}}.
\newblock \bibinfo{journal}{{\em VLDB\/}} \bibinfo{volume}{11},
  \bibinfo{number}{12} (\bibinfo{year}{1910}), \bibinfo{pages}{1910--1913}.
\newblock


\bibitem[\protect\citeauthoryear{Zhang, Liu, Zhou, Li, Xiao, Cheng, Xing, Wang,
  Cheng, Liu, Ran, and Li}{Zhang et~al\mbox{.}}{2019}]%
        {Zhang2019a}
\bibfield{author}{\bibinfo{person}{Ji Zhang}, \bibinfo{person}{Yu Liu},
  \bibinfo{person}{Ke Zhou}, \bibinfo{person}{Guoliang Li},
  \bibinfo{person}{Zhili Xiao}, \bibinfo{person}{Bin Cheng},
  \bibinfo{person}{Jiashu Xing}, \bibinfo{person}{Yangtao Wang},
  \bibinfo{person}{Tianheng Cheng}, \bibinfo{person}{Li Liu},
  \bibinfo{person}{Minwei Ran}, {and} \bibinfo{person}{Zekang Li}.}
  \bibinfo{year}{2019}\natexlab{}.
\newblock \showarticletitle{{An end-to-end automatic cloud database tuning
  system using deep reinforcement learning}}. In \bibinfo{booktitle}{{\em
  SIGMOD}}. \bibinfo{pages}{415--432}.
\newblock
\showISBNx{9781450356435}
\showISSN{07308078}
\showDOI{%
\url{https://doi.org/10.1145/3299869.3300085}}


\end{thebibliography}

\appendix

\section{Details on Experimental Setup}
\label{sec:parameter_space}

UDO, Simplified UDO, DDPG, SARSA, RL with Cache, and all UDO variants optimize the following system parameters for Postgres:

\begin{itemize}
    \item \verb|enable_bitmapscan|
    \item \verb|enable_mergejoin|
    \item \verb|force_parallel_mode|
    \item \verb|max_parallel_workers|
    \item \verb|max_parallel_workers_per_gather|
    \item \verb|enable_hashagg|
    \item \verb|enable_indexscan|
    \item \verb|enable_seqscan|
    \item \verb|enable_tidscan|
    \item \verb|enable_nestloop|
    \item \verb|maintenance_work_mem|
    \item \verb|effective_io_concurrency|
    \item \verb|temp_buffers|
    \item \verb|random_page_cost|
    \item \verb|seq_page_cost|
\end{itemize}

For MySQL, those baselines optimize the following parameters:

\begin{itemize}
    \item \verb|innodb_buffer_pool_size|
    \item \verb|max_heap_table_size|
    \item \verb|transaction_prealloc_size|
    \item \verb|transaction_alloc_block_size|
    \item \verb|table_open_cache|
    \item \verb|binlog_cache_size|
    \item \verb|read_buffer_size|
    \item \verb|join_buffer_size|
    \item \verb|thread_cache_size|
    \item \verb|max_length_for_sort_data|
    \item \verb|innodb_thread_concurrency|
\end{itemize}

For TPC-H, all of the aforementioned baselines use an evaluation script (i.e., benchmark metric) that sets per-query timeouts. Those per-query timeouts are initialized to the minimum between the respective query time with default settings and a constant (we use 30 seconds for TPC-H SF 10 and 6 seconds for TPC-H SF 1). Note that, while using those timeouts during optimization, we evaluate configurations without timeouts. I.e., the execution times shown in all plots of our experimental evaluation were generated without timeouts (this means we perform separate runs to evaluate the quality of generated configurations).

The baselines derived from the OtterTune demo~\cite{CMUDatabaseGroup2020} (DDPG++ and GP) use a different evaluation script for TPC-H that does not use per-query timeouts. Specifically for Postgres, we found that not having per-query timeouts leads to extremely long-running queries when changing optimizer cost constants (i.e., parameters \verb|random_page_cost| and \verb|seq_page_cost|). Changing those constants can mislead the query optimizer to select highly sub-optimal join orders. This prevented tuning progress. We therefore removed those two parameters from the search space of those two baselines. 

Also, while the evaluation script used by UDO and variants does not currently support parameters requiring a database server restart, OtterTune does support this functionality. We therefore added the \verb|shared_buffers| parameter to OtterTune's search space. This parameter requires a server restart. It is considered one of few most important tuning knobs by the baseline authors~\cite{VanAken2021}. Empirically, we found that the performance of DDPG++ and GP improves when adding this parameter. 

For MySQL, all baselines tune the same set of parameters.

\section{Additional Performance Metrics}

In Section~\ref{sec:experiments}, we have shown performance according to the metrics that UDO and the baselines optimize for. Often, there is a tradeoff with other metrics. To make those tradeoffs visible, we measure performance according to metrics not considered for optimization. We analyze the final configuration, selected by UDO and several baselines (during the experiments shown in Figures~\ref{fig:tpcc_benchmarks} and \ref{fig:tpch_benchmarks}). We compare according to a variety of performance metrics and report results in Table ~\ref{tab:postgresmoremetrics} and Table ~\ref{tab:mysqlmoremetrics}. We use Postgres and MySQL statistics to measure disk space consumption of indexes, other metrics are measured via the pg-activity tool\footnote{https://github.com/dalibo/pg\_activity}.



\begin{table}[]
\caption{\revision{Performance metrics for final configurations selected by different approaches on Postgres.}}
\label{tab:postgresmoremetrics}
\scalebox{0.75}{
\begin{tabular}{|l|l|l|l|l|}
\hline
Benchmark                  & Approach & Space(MB)  & Memory(GB) & CPU(\%) \\ \hline
\multirow{3}{*}{TPC-H SF10} & Dexter + PGTune   & 3383    & 15.54      & 2954.8  \\ \cline{2-4} 
                           & Dexter + OT GP     & 3383    & 9.26   & 2695.2 \\ \cline{2-5} 
                           & Dexter + OT DDPG++ & 3383    & 15.38      & 2696.7  \\ \cline{2-5} 
                           & UDO                & 2891    & 9.23         & 2497.1  \\ \hline
\multirow{3}{*}{TPC-C SF10} & Dexter + PGTune   & 2881.06   & 12.45     & 2434.2  \\ \cline{2-5} 
                           & Dexter + OT GP & 2881.06 & 11.98 &	2479.8 \\ \cline{2-5} 
                           & Dexter + OT DDPG++ & 2881.06   & 11.3     & 2461.2  \\ \cline{2-5} 
                           & UDO   & 2606.35    & 11.45     & 2486.78 \\ \hline
\end{tabular}
}
\end{table}

\begin{table}[]
\caption{\revision{Performance metrics for final configurations selected by different approaches on MySQL.}}
\label{tab:mysqlmoremetrics}
\scalebox{0.75}{
\begin{tabular}{|l|l|l|l|l|}
\hline
Benchmark                  & Approach & Space(MB)  & Memory(GB) & CPU(\%) \\ \hline
\multirow{3}{*}{TPC-H SF10} & Dexter + PGTune  &  9.05   &  34.7   & 141.4 \\ \cline{2-4} 
                           & Dexter + OT GP     &  9.05   & 38.8   & 153.4 \\ \cline{2-5} 
                           & Dexter + OT DDPG++ &  9.05   & 43.5  & 167.6   \\ \cline{2-5} 
                           & UDO                &  96   &   48.1      & 165.4  \\ \hline
\multirow{3}{*}{TPC-C SF10} & Dexter + PGTune   & 10.55  &  38.6    &  825 \\ \cline{2-5} 
                           & Dexter + OT GP     & 10.55  & 29.4     & 1598	 \\ \cline{2-5} 
                           & Dexter + OT DDPG++ & 10.55  &  33.5    & 1608 \\ \cline{2-5} 
                           & UDO   &  133.26 &   32.6   &  1798 \\ \hline
\end{tabular}
}
\end{table}

\balance

\end{document}